\documentclass[11pt]{article}
\usepackage[T1]{fontenc}
\usepackage{lmodern}
\usepackage{graphicx} %
\usepackage{fullpage}
\usepackage{amsmath,amsthm,amssymb}
\usepackage[linesnumbered,algoruled,boxed,lined]{algorithm2e}
\SetKwInOut{Input}{Input}
\SetKwInOut{Output}{Output}
\usepackage{algpseudocode}
\usepackage{mathtools}
\usepackage[%
    bookmarksnumbered,
    hypertexnames=false,
    pdfdisplaydoctitle,
    pdfusetitle,
    unicode
]{hyperref}
\usepackage{bm}
\usepackage{enumerate}
\usepackage{mathrsfs}
\usepackage{bbm}
\usepackage{float}
\usepackage{tikz}
\usepackage{csquotes}
\usepackage{breqn}
\usepackage[noabbrev,capitalize,nameinlink]{cleveref}
\usepackage[%
    doi=false,
    isbn=false,
    articlein=false,
    url=false,
    eprint=false,
    giveninits=true,
    maxbibnames=99,
    minalphanames=4,
    maxalphanames=4,
    natbib,
    sorting=nyt,
    style=ext-alphabetic
]{biblatex}

\DeclareFieldFormat[article,inbook,incollection,inproceedings,patent,thesis,unpublished]{titlecase:title}{\MakeSentenceCase*{#1}} %
\usepackage{subcaption}

\newtheorem{theorem}{Theorem}[section]
\newtheorem{lemma}[theorem]{Lemma}
\newtheorem{claim}[theorem]{Claim}
\newtheorem{corollary}[theorem]{Corollary}
\newtheorem{fact}{Fact}
\theoremstyle{definition}
\newtheorem{definition}[theorem]{Definition}

\numberwithin{equation}{section}
\numberwithin{figure}{section}

\foreach \x in {a,...,z}{%
	\expandafter\xdef\csname B\x\endcsname{\noexpand\mathbf{\x}}
}

\foreach \x in {A,...,Z}{%
	\expandafter\xdef\csname B\x\endcsname{\noexpand\mathbf{\x}}
}

\newcommand{\eps}{\varepsilon} 
\newcommand{\mc}{\mathcal}
\newcommand{\mb}{\mathbf}
\newcommand{\ol}{\overline}
\newcommand{\wt}{\widetilde}
\newcommand{\wh}{\widehat}
\newcommand{\R}{\mathbb{R}}
\newcommand{\Z}{\mathbb{Z}}
\newcommand{\bx}{\Bx}

\newcommand{\be}{\Be}
\newcommand{\caP}{\mathcal{P}}
\newcommand{\zeros}{\mathbf{0}}

\newcommand{\Lap}{\bm{\mc{L}}} 
\newcommand{\Tproj}{T_{\mathrm{proj}}} 

\DeclareMathOperator{\vol}{vol}
\DeclareMathOperator{\poly}{poly}
\DeclareMathOperator{\tr}{tr}
\DeclareMathOperator{\Diag}{Diag}
\DeclareMathOperator{\MVP}{MVP}
\DeclareMathOperator*{\E}{\mathbb{E}}
\DeclarePairedDelimiter{\norm}{\lVert}{\rVert}
\DeclarePairedDelimiter{\abs}{\lvert}{\rvert}
\DeclarePairedDelimiter{\inprod}{\langle}{\rangle}

\usepackage{xcolor}

\title{\bf $O(\log n)$-Approximation Algorithms for Bipartiteness Ratio}
\author{Tasuku Soma \\ The Institute of Statistical Mathematics \& RIKEN AIP \\ \url{soma@ism.ac.jp}
\and 
Mingquan Ye \\ National Institute of Informatics \& University of California, Santa Cruz \\ \url{mye26@ucsc.edu}
\and 
Yuichi Yoshida \\ National Institute of Informatics \\ \url{yyoshida@nii.ac.jp}
}
\date{\today}

\begin{document}

\maketitle
\begin{abstract}
We propose an $O(\log n)$-approximation algorithm for the bipartiteness ratio of undirected graphs introduced by \citet{Trevisan2012}, where $n$ is the number of vertices.
Our approach extends the cut-matching game framework for sparsest cut to the bipartiteness ratio, and requires only $\poly\log n$ many single-commodity undirected maximum flow computations.
Therefore, with the current fastest undirected max-flow algorithms, it runs in almost linear time.
Along the way, we introduce the concept of well-linkedness for skew-symmetric graphs and prove a novel characterization of bipartiteness ratio in terms of well-linkedness in an auxiliary skew-symmetric graph, which may be of independent interest.

As an application, we devise an $\wt{O}(mn)$-time algorithm for the minimum uncut problem: given a graph whose optimal cut leaves an $\eta$ fraction of edges uncut, we find a cut that leaves only an $O(\log n \log(1/\eta)) \cdot \eta$ fraction of edges uncut, where $m$ is the number of edges.

Finally, we propose a directed analogue of the bipartiteness ratio, and we give a polynomial-time algorithm that achieves an $O(\log n)$ approximation for this measure via a directed Leighton--Rao-style embedding. 
We also propose an algorithm for the minimum directed uncut problem with a guarantee similar to that for the minimum uncut problem.
\end{abstract}

\section{Introduction}\label{sec:intro}

Let $G = (V, E, w)$ be an undirected graph with $n$ vertices, $m$ edges, and a positive edge weight $w: E \to \Z_{++}$, where $\Z_{++}$ is a set of positive integers. 
The \emph{(normalized) Laplacian matrix} of $G$ is given by $\BI_n - \BD^{-1/2} \BA \BD^{-1/2}$, where $\BI_n \in \R^{n \times n}$ is the $n \times n$ identity matrix, $\BA \in \R^{n \times n}$ is the weighted adjacency matrix of $G$, and $\BD \in \R^{n \times n}$ is the diagonal matrix with $\BD_{ii}$ being equal to the weighted degree of vertex $i$.
The Laplacian matrix is symmetric and positive semidefinite, and its eigenvalues satisfy $0 = \lambda_1 \leq \cdots \leq \lambda_n \leq 2$.
A classical result in spectral graph theory states that $G$ is bipartite if and only if $\lambda_n = 2$.
\citet{Trevisan2012} proved a quantitative version of this result.
For a nonzero vector $\Bx \in \{0, \pm 1\}^V$, let
\[
    \beta(\Bx) := \frac{\sum_{e = (i, j) \in E} w(e) \cdot \abs{x_i + x_j}}{\sum_{i \in V} \deg(i) \cdot \abs{x_i}}, 
\]
where $\deg(i) = \sum_{(i, j) \in E} w(i, j)$ is the weighted degree of vertex $i$.
The \emph{bipartiteness ratio} of $G$ is then defined by 
\begin{align}\label{eq:bip_rat}
    \beta(G) := \min_{\Bx \in \{0, \pm 1\}^V \setminus \{\zeros \}} \beta(\Bx). 
\end{align}
Since each non-zero $\{0,\pm 1\}$-vector $\bx$ corresponds to a tripartition $(L, R, Z)$ of $V$ such that $L = \{i \in V \mid x_i = 1\}$, $R = \{j \in V \mid x_j = -1\}$, and $Z = \{k \in V \mid x_k = 0\}$, we can represent 
\begin{align*}
    \beta(G) = \min_{\text{$(L, R, Z)$: tripartition of $V$}} \frac{2w(E(L)) + 2w(E(R)) + w(E(L \cup R, Z))}{\vol(L \cup R)},
\end{align*} 
where $E(L)$ (resp.,~$E(R)$) is the set of edges whose endpoints are within $L$ (resp.,~$R$), and $E(L \cup R, Z)$ is the set of edges connecting $L \cup R$ and $Z$, and $\vol(L \cup R) = \sum_{i \in L \cup R} \deg(i)$ is the volume of $L \cup R$.
Obviously, $\beta(G) = 0$ if and only if $G$ is bipartite.
\citet{Trevisan2012} showed\footnote{Although Trevisan~\cite{Trevisan2012} states the Cheeger‐type inequality for unweighted graphs, the same argument carries over directly to the weighted case with only minor adjustments.} that the bipartiteness ratio is closely related to the largest eigenvalue $\lambda_n$ of the Laplacian matrix, specifically,  
\[
    \frac{2 - \lambda_n}{2} \le \beta(G) \le \sqrt{2 (2 - \lambda_n)}. 
\]
Furthermore, Trevisan also present a simple algorithm that finds a nonzero vector $\bx \in \{0, \pm 1\}^V$ such that $\beta(\bx) \leq \sqrt{2(2 - \lambda_n)}$ given an eigenvector corresponding to $\lambda_n$.

Trevisan's inequality can be regarded as an analogue of the \emph{Cheeger's inequality}~\cite{Alon1985,Alon1986}, which relates the second smallest eigenvalue with the \emph{conductance} of graphs.
For a vertex subset $\emptyset \subsetneq S \subsetneq V$, let
\[
    \phi(S) := \frac{w(E(S, \ol{S}))}{\min\{\vol(S), \vol(\ol{S})\}},
\]
where $E(S, \ol{S})$ is the set of edges connecting $S$ and $\ol{S}$.
Then the conductance of $G$ is defined as
\[
    \phi(G) := \min_{\emptyset \subsetneq S \subsetneq V} \phi(S).
\]
The Cheeger's inequality states that
\[
    \frac{\lambda_2}{2} \leq \phi(G) \leq \sqrt{2\lambda_2}.
\]
Furthermore, there is a simple algorithm that finds $S$ such that $\phi(S) \leq \sqrt{2\lambda_2}$ given an eigenvector corresponding to $\lambda_2$.

Just as the Cheeger's inequality captures a graph’s expansion properties through the second smallest eigenvalue of its Laplacian, the bipartiteness ratio quantifies a graph’s deviation from bipartiteness via the largest eigenvalue of the normalized Laplacian. 
Indeed, Trevisan's inequality (and its variants) is sometimes called the \emph{dual Cheeger's inequality}~\cite{Bauer2013,pokharanakar25}.
As an application of his inequality, \citet{Trevisan2012} showed a purely spectral algorithm for max cut with a nontrivial approximation ratio better than $1/2$.
Recently, bipartiteness ratio has found applications in network analysis in a similar fashion to graph conductance and sparsest cut~\cite{xog20,al20,np22}.\footnote{Indeed \cite{xog20,al20,np22} studied even a more general network model called \emph{signed graphs} and analyzed spectral clustering. Since the main motivation of the present paper is in approximation algorithms for the bipartiteness ratio, we focus on ordinary undirected graphs rather than signed graphs.}
Despite these rich connections and their algorithmic promise, obtaining nontrivial approximation guarantees for the bipartiteness ratio remains an open challenge: no polynomial-time approximation algorithm is currently known. Developing such an algorithm would yield new spectral and combinatorial tools both for classical maximum cut derivatives and for the burgeoning field of network analysis. 

\subsection{Our Contributions}
We present the first $O(\log n)$-approximation algorithm for the bipartiteness ratio of undirected graphs.
More precisely, we study the following \emph{$b$-bipartiteness ratio}, which generalizes the original bipartiteness ratio~\cite{Trevisan2012}.
Let $b: V \to \Z_{++}$ be a positive vertex weight.
For a vector $\bx \in \{0, \pm 1\}^V \setminus \{\mb{0}\}$, let   
\[
    \beta_b(\bx) := \frac{\sum_{(i, j) \in E} w(i, j) \cdot |x_i + x_j|}{\sum_{i \in V} b(i) \cdot |x_i|}.
\]
Then, we define the $b$-bipartiteness ratio of $G$ by $\beta_b(G) := \min_{\bx \in \{0, \pm 1\}^V \setminus \{\zeros\}} \beta_b(\bx)$.
Note that the original bipartiteness ratio $\beta(G)$ (see \eqref{eq:bip_rat}) is the special case of $\beta_b(G)$ where $b(i) = \deg(i)$ ($i \in V$).

Here is our main result.
\begin{theorem}[informal version of \Cref{thm:cut-matching}]
    There is a randomized $O(\log n)$-approximation algorithm for the $b$-bipartiteness ratio of an undirected graph. 
    That is, the algorithm finds a nonzero vector $\bx \in \{0, \pm 1\}^V$ such that $\beta_b(\bx) \leq O(\log n) \cdot \beta_b(G)$ with probability at least $1 - 1/\poly(n)$.
    The time complexity is $O(\log (w(E) \cdot b(V)) \cdot \log^3 n \cdot \max\{\log^2 n, \log b(V)\} \cdot \min\{b(V), n^2\})$ arithmetic operations and $O(\log (w(E) \cdot b(V)) \cdot \log^2 n)$ single-commodity max-flow computations on an auxiliary undirected graph of size $O(m+n)$.
\end{theorem}

By employing the almost-linear–time algorithm for exact maximum flow~\cite{cklpps22}, our algorithm runs in $\wt{O}(\min\{b(V), n^2\} + m)$ time. In the case of the original bipartiteness ratio, where $b(V) = O(m)$, this simplifies to a total running time of $\wt{O}(m)$.

\subsection{Application to Min Uncut}
The \emph{maximum cut} and \emph{minimum uncut} problems are extensively studied classical combinatorial optimization problems. In fact, they are closely related to the bipartiteness ratio problem. 
To see that, recall that each non-zero vector $\Bx \in \{0, \pm 1\}^V$ corresponds to a partition of $V$ such that $L = \{i \mid x_i = 1\}$, $R = \{j \mid x_j = -1\}$, and $Z = \{k \mid x_k = 0\}$. Suppose $Z \subseteq V$ is given, that is, we set the subvector $\Bx_Z$ of $\Bx$ corresponding to $Z$ to zero, then we have 
\begin{equation}
\label{eq:bip_maxcut}
    \begin{aligned}
&~\min_{\Bx \in \{0, \pm 1\}^V \setminus \{\mb{0}\},\ \Bx_Z 
= \mb{0}} \frac{\sum_{(u, v) \in E} w(u, v) \cdot |x_u + x_v|}{\sum_{u \in V} b(u) \cdot |x_u|} \\
=&~ \min_{L \cup R = V \setminus Z,\ L \cap R = \emptyset} \frac{2 \cdot w(E(L)) + 2 \cdot w(E(R)) + w(E(L \cup R, Z))}{b(L \cup R)} \\
=&~ \min_{L \cup R = V \setminus Z,\ L \cap R = \emptyset} 1 - \frac{2 \cdot w(E(L, R))}{b(L \cup R)} \\
=&~ 1 - 2 \cdot \max_{L \cup R = V \setminus Z,\ L \cap R = \emptyset} \frac{w(E(L, R))}{b(L \cup R)},  
\end{aligned} 
\end{equation}
which is reduced to computing the maximum cut of the subgraph induced by the vertex subset $V \setminus Z$ since $Z$ is given.

Likewise, given an approximate algorithm for bipartiteness ratio, we can develop some approximate algorithm for max cut~\cite{kllot13,Trevisan2012}. Given a non-zero vector $\Bx \in \{0, \pm 1\}^V$ corresponding to the partition $V = L \cup R \cup Z$, we have 
\[
\beta(\Bx) = \frac{2 \cdot w(E(L)) + 2 \cdot w(E(R)) + w(E(L \cup R, Z))}{\vol(L \cup R)} = 1 - 2 \cdot \frac{w(E(L, R))}{\vol(L \cup R)}.  
\]
If we have an upper bound for $\beta(\Bx)$, that naturally gives rise to a lower bound for $\frac{w(E(L, R))}{\vol(L \cup R)}$. Consequently, we can lower bound the fraction of cutting edges in the subgraph induced by the vertex subset $L \cup R$, and continue this process on the remaining subgraph of $Z$. The following is our result for min uncut.
\begin{theorem}\label{thm:min-uncut-intro} 
Given a graph whose minimum uncut fraction is $\eta$, there exists a randomized algorithm that returns a cut leaving at most an $O(\log n \log (1/\eta)) \cdot \eta$ fraction of uncut edges with probability at least $1 - 1/\poly(n)$ in $\wt{O}(mn)$ time.  
\end{theorem}

\Cref{tab:maxcut} compares our result with the state-of-the-art min uncut algorithms with the same style guarantees~\cite{Trevisan2012,kllot13,gvy93,acmm05}. Compared with \cite{gs11,kllot13}, our approximation guarantee is independent of the eigenvalue of the normalized Laplacian matrix. For \cite{Trevisan2012}, although our method is similar to it, our algorithm depends roughly linearly on $\eta$, while \cite{Trevisan2012}'s work depends on $\sqrt{\eta}$. 
Moreover, compared with the linear programming (LP)~\cite{cls21} and semidefinite programming (SDP)~\cite{hjstz22} based approximation~\cite{gvy93,acmm05}, although our method performs worse on the approximation ratio, the running time is $\wt{O}(mn)$, which significantly outperforms their running time. 

\begin{table}[t]
  \centering
  \caption{Summary of known min-uncut algorithms and our work. Here $\alpha_k$ denotes the $k$-th smallest eigenvalue of the matrix $2\BI_n - \Lap$ and $\lambda_{n-k}$ denotes the $(n-k)$-th smallest eigenvalue of $\Lap$. The parameter $\eps \in (0,1)$ in \cite{gs11} is an arbitrary fixed constant.}\label{tab:maxcut}
  \begin{tabular}{lll}
    Reference & Uncut fraction  & Time complexity \\ \hline\hline
    \cite{Trevisan2012} & $O(\sqrt{\eta})$ & Spectral decomposition \\ \hline
    \cite{kllot13}    & $O(\frac{k}{\alpha_k}\log\frac{\alpha_k}{k\eta}) \cdot \eta$  & Spectral decomposition \\ \hline
    \cite{gs11} & $\frac{1+\eps}{\lambda_{n-k}} \cdot \eta$  & $2^{O(k/\varepsilon^3)}\,n^{O(1/\varepsilon)}$ \\ \hline
    \cite{gvy93} & $O(\log n) \cdot \eta$ & $\wt{O}(m^\omega)$~\cite{jswz21} \\ \hline
    \cite{acmm05} & $O(\sqrt{\log n}) \cdot \eta$ & $\wt{O}(m^\omega)$~\cite{hjstz22} \\ \hline
    \bf This work (\Cref{thm:max-cut}) & $O(\log n \log(1/\eta)) \cdot \eta$ & $\wt{O}(mn)$ 
    \\ \hline 
    \end{tabular}
\end{table}

\subsection{Directed Bipartiteness Ratio}

Let $G = (V, E)$ be a directed graph.
We say that $G$ is \emph{(strongly) bipartite} if there exists a bipartition $(L, R)$ of $V$ such that all arcs in $E$ go from $L$ to $R$.
Inspired by the bipartiteness ratio of undirected graphs, we define the \emph{directed $b$-bipartiteness ratio} of $G$ as follows.
First, for $\Bx \in \R^V$, let
\begin{align} \label{def:dir_beta_b_x}
    \beta_b(\Bx) = \frac{\sum_{(i, j) \in E} \psi_{ij}(\Bx)}{\sum_{i \in V} b(i) \abs{x_i}},
    \quad \text{where } 
    \psi_{ij}(\Bx) := \begin{cases}
        \abs{x_i + x_j} & \text{if $x_i \geq x_j$}, \\
        \abs{x_i} + \abs{x_j} & \text{otherwise}.
    \end{cases}
\end{align} 
Then, the directed $b$-bipartiteness ratio of $G$ is defined as
\[
    \beta_b(G) = \min_{\Bx \in \{0, \pm 1\}^V \setminus \{\mb 0\}} \beta_b(\Bx).
\]
Intuitively, we look for a large witness of bipartiteness: a tripartition $(L,R,Z)$ with $L$ and $R$ playing the roles of the two sides, whose violations of the forward orientation are few compared with the total $b$-weight of vertices that participate.
When $\Bx \in \{0,\pm 1\}^V$ encodes such a tripartition $(L, R, Z)$ as above, 
\[
    \beta_b(\Bx) = \frac{2 \cdot (w(E(R, L)) + w(E(L)) + w(E(R)) + w(E(L \cup R, Z)) + w(E(Z, L \cup R))}{b(L \cup R)},
\]
where $E(R, L)$ denotes the set of arcs from $R$ to $L$ and $E(L)$ denotes the set of arcs within $L$.
Thus a small ratio certifies the existence of a large almost strongly bipartite subgraph, and the ratio becomes zero exactly when $G$ contains a nontrivial strongly bipartite subgraph. Similar to \eqref{eq:bip_maxcut}, the directed bipartiteness ratio has a close relation with directed min uncut. For a tripartition $(L, R, Z)$ of $\Bx \in \{0, \pm 1\}^V$ with $Z \subset V$ being fixed, we have 
\begin{align} \label{eq:bip_minuncut} 
    \min_{L \cup R = V \setminus Z,\ L \cap R = \emptyset} \beta_b (\Bx) = 1 - 2 \cdot \max_{L \cup R = V \setminus Z,\ L \cap R = \emptyset} \frac{w(E(L, R))}{b(L \cup R)},   
\end{align}
which is equivalent to solving the directed min uncut of subgraph $G[L \cup R]$. Moreover, we can verify that the function $\psi_{ij}$ is \emph{bisubmodular}~\cite{fujishige2005bisubmodular} by an exhaustive case analysis.

We show that we can approximate $\beta_b(G)$ within a factor of $O(\log n)$ in polynomial time:
\begin{theorem}\label{thm:directed-bipartiteness-ratio}
    There is a randomized polynomial-time $O(\log n)$-approximation algorithm for the directed $b$-bipartiteness ratio of a directed graph. 
    That is, the algorithm finds a nonzero vector $\bx \in \{0, \pm 1\}^V$ such that $\beta_b(\bx) \leq O(\log n) \cdot \beta_b(G)$ with probability at least $1 - 1/\poly(n)$.
\end{theorem}

The \emph{minimum directed uncut} problem is the directed analogue of minimum uncut: given a directed graph $G = (V, E)$ and a bipartition $(S, \bar{S})$ of the vertex set, an edge is uncut if it does not go from $S$ to $\bar{S}$ (that is, it stays inside $S$ or $\bar{S}$, or points from $\bar{S}$ to $S$). The goal is to choose $(S, \bar{S})$ that minimizes the total weight of such uncut edges, or equivalently maximizes the weight of the directed cut $E(S, \bar{S})$.
We have the following theorem using \Cref{thm:directed-bipartiteness-ratio}.
\begin{theorem}\label{thm:min-directed-uncut-intro}
Given a graph with the fraction of the minimum directed uncut being $\eta$, there is a randomized polynomial-time algorithm for finding a cut that leaves the fraction of edges $O(\log n \log (1/\eta)) \cdot \eta$ with probability at least $1 - 1/\poly(n)$.
\end{theorem} 

The best known approximation ratio of minimum directed uncut is $O(\sqrt{\log n})$ based on SDP~\cite{acmm05}, so our approximation ratio is weaker.
The primary advantage of our approach is its simplicity: we obtain it via a recursive reduction to approximating the directed bipartiteness ratio, mirroring the undirected setting.
If we could approximate the directed bipartiteness ratio in $\wt{O}(m)$ time, as in the undirected case, then our algorithm would run in $\wt{O}(mn)$ time, outperforming \cite{acmm05}.
Whether such a fast approximation exists is currently unclear, and we leave this as an open problem.

\subsection{Our Technique}
We briefly describe our techniques here.

\subsubsection*{Cut-matching game for bipartiteness ratio}
Our approach for the approximation algorithm for bipartiteness ratio is extending the \emph{cut-matching game} framework~\cite{krv06}, which is originally designed for sparsest cut (i.e., $b \equiv 1$), to bipartiteness ratio.

Let us first review the original cut-matching game for sparse cut.
Suppose that we want to check whether $\phi_b(G) \geq 1$ or not.
The key fact of the cut-matching game is that $\phi_b(G) \geq 1$ if and only if $G$ is \emph{well-linked}, i.e., for any disjoint vertex subsets $A, B \subseteq V$ with $\abs{A} = \abs{B}$, there exists an $A$--$B$ flow in $G$ such that it satisfies edge capacity $w$ and unit flow goes out from and in to every vertex in $A$ and $B$, respectively; see, e.g., \cite{ChekuriLect}.
Such a flow is said to be \emph{saturating}.
The cut-matching game is a repeated game of two players, the cut player and the matching player.
Let $H$ be an empty multigraph on $V$.
In each round, the cut player generates a bipartition $(S, \ol{S})$ of the vertex set $V$.
Without loss of generality, we assume that $\abs{S} \leq \abs{\ol{S}}$.
If $S$ is not well-linked to some subset in $\ol{S}$---which can be checked in a single-commodity flow computation---then the game ends; by the above characterization, $\phi_b(G) < 1$.
Indeed, we can even find a sparse cut $S$ as well by finding a minimum cut.
If $S$ is well-linked, there must exist a flow that saturates $S$ by definition.
The matching player finds such a flow and adds its demand graph to $H$ (as a multigraph).\footnote{For the unweighted case (i.e., $w \equiv 1$), the demand graph is simply a matching between $S$ and $\ol{S}$ of size $\abs{S}$, hence the name of the matching player.}
The game ends once $H$ becomes an $O(1)$-expander, i.e., $\phi_b(H) \geq \Omega(1)$.
Suppose that the game ends after $T$ rounds by finding an expander $H$.
Then, since $H$ is embeddable to $G$ with congestion $O(T)$, we have $\phi_b(G) \geq \Omega(1/T)$.
Thus, we achieve an $O(T)$-approximation for sparsest cut.
\citet{krv06} showed that there exists a randomized strategy of the cut player such that the game ends after $T = O(\log^2 n)$ rounds with high probability.
This is later improved to $O(\log n)$-approximation by \citet{osvv08}.
Furthermore, \citet{Arora2016} provides a very systematic interpretation and analysis of the cut-matching game with \emph{matrix multiplicative weight update (MMWU)} and proves $T = O(\sqrt{\log n})$, which is the current best approximation ratio.

To design a cut-matching game for bipartiteness ratio, we first prove an analogous characterization of bipartiteness ratio in terms of flows in an \emph{auxiliary graph} $G'$.
Let $V^+$ and $V^-$ be the disjoint copies of $V$. 
We denote the copies of a vertex $i$ in $V^+$ and $V^-$ by $i^+$ and $i^-$, respectively.
Let $E'$ be the set of edges between $V^+$ and $V^-$ such that for each edge $(i, j) \in E$, there are two corresponding edges $(i^+, j^-) \in E'$ and $(i^-, j^+) \in E'$.
Let $G' = (V^+ \cup V^-, E')$ be the resulting undirected graph. Note that $G'$ is bipartite by construction. For $X \subseteq V$, denote by $X^+$ (resp.~$X^-$) the corresponding set of $X$ in $V^+$ (resp.~$V^-$). For $G$ and $G'$, see~\cref{fig:ori_bip} for an illustration. 
Given a vertex weight $b$ on $V$, we can naturally induce a vertex weight $b'$ on $V'$ such that $b'(u) = b(i)$ for $i \in V$ and $u \in \{i^+, i^-\}$.
With a slight abuse of notation, we also denote the induced vertex weights on $V'$ by $b$, and similarly, the induced edge weights on $E'$ by $w$.

\begin{figure}[t] 
    \centering
    \subcaptionbox{\label{fig:ori}}[.45\linewidth]{\includegraphics[scale=0.8]{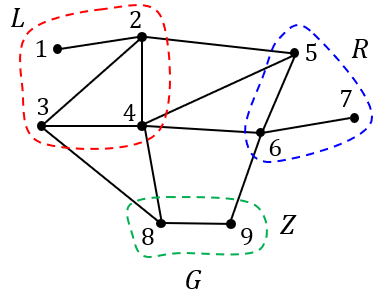}} %
    \subcaptionbox{\label{fig:bip}}[.45\linewidth]{\includegraphics[scale=0.8]{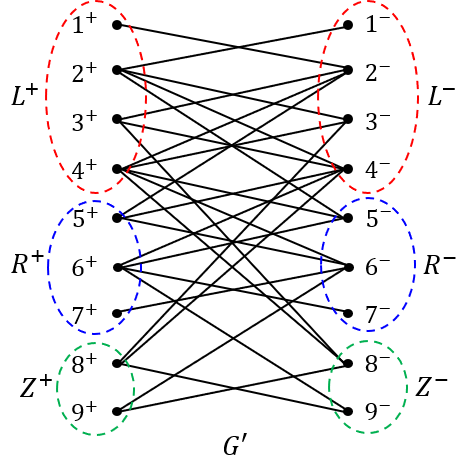}}
    \captionsetup{subrefformat=parens}
    \caption{\subref{fig:ori} the original graph $G = (V, E)$. Given a vector $\Bx = [1, 1, 1, 1, -1, -1, -1, 0, 0]^\top$, we have a corresponding partition of $V = L \cup R \cup Z$, where $L = \{1, 2, 3, 4\}$, $R = \{5, 6, 7\}$, and $Z = \{8, 9\}$. \subref{fig:bip} the corresponding bipartite graph $G' = (V^+ \cup V^-, E')$, where each edge $(u, v) \in E$ correponds to two edges $(u^+, v^-)$ and $(v^+, u^-)$ in $E'$. Moreover, the three subsets $L, R$, and $Z$ of $V$ correspond to the subsets $L^+, L^-$, $R^+, R^-$, and $Z^+, Z^-$ of $V'$, respectively.}
    \label{fig:ori_bip}
\end{figure}

The auxiliary graph reveals that bipartiteness ratio has a structure quite similar to sparsest cut.
For any $\bx \in \{0, \pm 1\}^V \setminus \{\mb{0}\}$ corresponding to tripartition $(L, R, Z)$ of $V$, let $S := L^+ \cup R^-$ and $\ol{S} := V' \setminus S$.
Then we can show  
\[
    \beta_b(\bx) = \frac{w(E'(S, \ol{S}))}{b(S)},  
\]
and thus we can represent the $b$-bipartiteness ratio as
\[
    \beta_b(G) = \min_{S = L^+ \cup R^-,\ \textnormal{disjoint } L, R \subseteq V} \frac{w(E'(S, \ol{S}))}{b(S)}.
\]
This formula is the same as generalized sparsest cut~\eqref{eq:gen_sparsest_cut}, except for the additional constraint on $S$.
Note that if $S = L^+ \cup R^-$, then $\ol{S} \supseteq L^- \cup R^+$ and therefore $\min\{b(S), b(\ol{S})\} = b(S)$.

Furthermore, the auxiliary graph $G'$ has the following symmetry: $(i^+, j^-) \in E'$ if and only if $(j^+, i^-) \in E'$.
This is a special case of \emph{skew-symmetric} graphs~\cite{Goldberg1996}.
Taking this symmetry into account, we define that $G'$ is well-linked if any \emph{symmetric} source-sink pair $A = L^+ \cup R^-$, $B = L^- \cup R^+$ for disjoint $L, R \subseteq V$ admits a saturating flow in $G'$; see \Cref{def:well-linked} for the formal definition.
Then, we show that one can characterize the bipartiteness ratio in terms of well-linkedness (\Cref{thm:linked}):
$\beta_b(G) \geq 1$ if and only if $G'$ is well-linked.
Again, this is parallel to the aforementioned characterization of sparsest cut in terms of well-linkedness.

Using our novel characterization of bipartiteness ratio, we propose the following cut-matching game for bipartiteness ratio.
Again, let $H$ be an empty multigraph on $V$.
In each round, the cut player generates a tripartition $(L, R, Z)$ of $V$.
Let $A = L^+ \cup R^-$ and $B = L^- \cup R^+$ be the corresponding symmetric source-sink pair in $G'$.
If $(A, B)$ is not well-linked in $G'$, then the game ends; we know $\beta_b(G) < 1$.
Again, we can even find $\bx \in \{0, \pm 1\}^V$ such that $\beta_b(\bx) \leq 1$ by finding a minimum cut.
Otherwise, there exists a saturating flow in $G'$ from $A$ to $B$, which the matching player finds.
Then, we add to $H$ a multigraph on $V$ induced from the demand graph of the flow.
The game ends if $H$ satisfies $\beta_b(H) \geq \Omega(1)$. 
Using a MMWU analysis similar to \cite{Arora2016}, we present a randomized strategy of the cut player such that the game ends after $T = O(\log n)$ rounds with high probability.
Since $H$ is embeddable to $G$ with congestion $O(T)$, we achieve the desired $O(\log n)$-approximation.

Each iteration of the cut matching game amounts to finding a Gram decomposition of a certain positive definite matrix provided by MMWU and a single-commodity max-flow on the auxiliary graph $G'$.
Roughly speaking, the former task is to compute a Gram decomposition of the matrix exponential of an $n \times n$ real symmetric matrix.
Although the standard method takes $O(n^3)$ time, there is a technique for computing \emph{approximate} Gram decompositions much more efficiently~\cite{Arora2016,Lau2024}, which we can adopt to our setting.

Finally, we remark that the same construction of $G'$ was used in the context of approximating min uncut~\cite{acmm05}, which is to find a partition $(L, R)$ (i.e., $Z = \emptyset$) minimizing uncut $w(E(L) \dot\cup E(R))$.
Note that since $Z = \emptyset$, minimizing uncut is equivalent to minimizing the ratio objective $\beta_b$.
We generalize their result to show that the same construction works for bipartiteness ratio, which is to find a tripartition $(L, R, Z)$ minimizing $\beta_b$.

\subsubsection*{Application to min uncut}
Building on our approximation algorithm for bipartiteness ratio, we propose an approximation algorithm for the minimum uncut problem. 
Our method is conceptually similar to those in \cite{kllot13,Trevisan2012}, which showed a close connection between min uncut and bipartiteness ratio. 
Briefly, we invoke \cref{alg:cut-matching}, which returns a partition of the vertex set $V=L \cup R \cup Z$, and then recursively apply this process to the subgraph induced by the vertex subset $Z$, where the entries of $\Bx_Z$ are zero. 
Since the size of vertex set decreases in every iteration, this method would take $\wt{O}(mn)$ time.

\subsubsection*{Directed bipartiteness ratio}

To establish \Cref{thm:directed-bipartiteness-ratio}, we follow a two-step blueprint in the spirit of the Leighton--Rao algorithm.
We first build a skew-symmetric auxiliary graph $G'$ of the input digraph by replacing every arc with a constant-size gadget.
This lets us represent the directed $b$-bipartiteness ratio as the minimum ratio-cut value over consistent subsets of $G'$, generalizing the undirected cut formulation.
We then solve a directed semimetric relaxation whose optimal value upper-bounds $\beta_b(G)$.  
Applying the directed $\ell_1$ weak-embedding theorem~\cite{charikar2006directed}, we round the solution by sampling a random directed cut in the embedded space and projecting it to a consistent cut; this yields the desired $O(\log n)$ approximation.

Applying the same recursion as in our min uncut algorithm, but now equipped with the directed bipartiteness ratio approximation, yields the algorithm for minimum directed uncut.

\subsection{Related work} 

Trevisan~\cite{Trevisan2012} defined the formulation of the bipartiteness ratio for an undirected, unweighted graph $G$, and proved an analog of Cheeger’s inequality involving $\lambda_n$, the largest eigenvalue of the normalized Laplacian of $G$. 
Additionally, Trevisan proposed an algorithm that given a graph $G$ with the fraction of the maximum cut edges being $1 - \eta$, returns a vector $\Bx$ in polynomial time such that $\beta(\Bx) \le 2 \sqrt{\eta}$. 
By recursively applying this algorithm to the subgraph induced by the zero entries of $\Bx$, one can find a cut that cuts at least a $1 - \Theta(\sqrt{\eta})$ fraction of the edges. 
Recently, Pokharanakar~\cite{pokharanakar25} extended the concept of the bipartiteness ratio from finite graphs to the setting of graphons. In this work, the author formulated an appropriate definition of bipartiteness ratio for graphons and proved a Cheeger-type inequality that relates this ratio to the top of the spectrum of the associated graphon Laplacian. 

For a positive vertex weight $b: V \to \Z_{++}$, the \emph{generalized sparsest cut} problem\footnote{The most general form of generalized sparsest cut allows arbitrary demand graphs; see \cite[Section~19.3]{KorteVygen2018}. Here, we only consider the case such that the demand graph is a complete graph and $b$ is arbitrary, which is relevant to bipartiteness ratio.} is to compute
\begin{align}\label{eq:gen_sparsest_cut}
    \phi_b(G) := \min_{\emptyset \subsetneq S \subsetneq V} \frac{w(E(S, \ol{S}))}{\min\{b(S), b(\ol{S})\}},
\end{align}
where $b(S) := \sum_{i\in S}b(i)$.
The \emph{conductance} corresponds to the case where $b(i) = \deg(i)$ for all $i \in V$.
For $b(i) = 1$ for all $i \in V$, the problem is known as the \emph{sparsest cut} problem.
For $w(e) = 1$ ($e \in E$) and $b(i) = 1$ ($i \in V$), $\phi_b(G)$ is called the \emph{edge expansion} of $G$. Computing the conductance or (generalized) sparsest cut of a graph is NP-hard~\cite{Shahrokhi1990}, so research has focused on efficient approximation algorithms. Khandekar, Rao, and Vazirani~\cite{krv06} introduced a purely combinatorial $O(\log^2 n)$‐approximation via the cut–matching game. Leighton and Rao~\cite{Leighton1999} achieved an $O(\log n)$‐approximation by solving a multicommodity flow relaxation, and Orecchia et al.~\cite{osvv08} showed the same guarantee can be obtained with a sequence of single‐commodity max‐flow calls. SDP relaxations yield better ratios: Arora, Rao, and Vazirani~\cite{Arora2009} gave an $O(\sqrt{\log n})$‐approximation by rounding an SDP. 
Sherman~\cite{Sherman2009} achieved an $O(\sqrt{\log (n) / \eps})$-approximation with $\tilde O(n^{\eps})$ max‐flow computations for any $\eps > 0$. Furthermore, Arora et al.~\cite{ahk10,Arora2016} obtained the $O(\sqrt{\log n})$ ratio in $\wt{O}(n^2)$ time. In \cite{ahk10}, they did so by efficiently constructing expander flows, whereas in~\cite{Arora2016}, they introduced a general primal–dual framework for solving SDPs via MMWU.

The cut-matching game, originally proposed by Khandekar et al.~\cite{krv06}, provides a fast combinatorial framework for approximating the sparsest cut in undirected graphs through flow-based techniques. In their work~\cite{krv06}, they designed a cut player strategy capable of constructing a graph with edge expansion $\Omega(1)$ within $O(\log^2 n)$ rounds. 
This was later improved by Orecchia et al.~\cite{osvv08}, which developed a cut player strategy that achieves edge expansion $\Omega(\log n)$ in the same number of rounds. 
Moreover, Louis~\cite{louis10} introduced a cut player strategy for cut-matching game on directed graphs, and leveraged it to develop an $O(\log^2 n)$-approximate algorithm for directed sparsest cut using $O(\log^2 n)$ max-flow computations. 
Subsequently, Lau et al.~\cite{Lau2024} improved upon the cut-matching game framework for directed graphs introduced in~\cite{louis10}, leading to an 
$O(\log n)$-approximate algorithm for directed edge expansion that runs in almost linear time. 
Additionally, it is worth noting that the cut-matching game has evolved into a versatile algorithmic primitive, finding applications in a variety of domains such as edge-disjoint paths~\cite{andrews10,chuzhoy12,cl12}, dynamic graph algorithms~\cite{ns17,ck19,bgs20,cglnps20,bgs21,grst21,chuzhoy23,cz23}, hypergraph ratio cuts~\cite{chen2025submodular,veldt23}, expander decompositions~\cite{sw19,cs20,ls22,hht24,cmm25,hhg25}, hierarchical decomposition~\cite{rst14,grst21}, network flows~\cite{peng16,hhlrs24}, etc.

For the maximum cut problem, which is equivalent to the minimum uncut problem, a simple deterministic algorithm achieves a $1/2$-approximation in polynomial time. Under the \emph{unique games conjecture}, the best approximation ratio is the Goemans–Williamson bound of $\alpha_{\mathrm{GW}}\approx0.878$~\cite{gw94}, obtained via an SDP relaxation followed by randomized hyperplane rounding. 
Suppose a maximum cut leaves an $\eta$ fraction of edges. 
Trevisan~\cite{Trevisan2012} gave a purely spectral algorithm that, in this regime, finds a cut leaving an $O(\sqrt{\eta})$ fraction of edges. 
Kwok et al.~\cite{kllot13} generalized this approach by considering higher-order eigenvalues: 
if $\alpha_k$ denotes the $k$-th smallest eigenvalue of the matrix $2\BI_n - \Lap$, then their algorithm produces a cut that leaves
$O(\frac{k}{\alpha_k} \log \frac{\alpha_k}{k \eta}) \cdot \eta$ fraction of edges.
Building on SDP hierarchies, Guruswami and Sinop~\cite{gs11} showed that, for any $\varepsilon\in(0,1)$, one can achieve a fraction of uncut edges $O(\frac{1+\eps}{\lambda_{n-k}}) \cdot \eta$ in time $2^{O(k/\varepsilon^3)}\,n^{O(1/\varepsilon)}$, where $\lambda_{n-k}$ is the $(n-k)$-th smallest eigenvalue of the normalized Laplacian and the rounding is performed on an SDP from the Lasserre hierarchy~\cite{lasserre02}. 
Since the minimum uncut is equivalent to the maximum cut under complementarity, one may also leverage approximation algorithms for the former. 
Garg et al.~\cite{gvy93} reduced min-uncut to the minimum multicut problem, yielding an $O(\log n)$-approximation for min-uncut. 
Arora et al.~\cite{acmm05} later improved this to an $O(\sqrt{\log n})$-approximation by solving SDP relaxations recursively.

It is a natural progression to generalize algorithms developed for undirected graphs to directed settings. Hajiaghayi and Räcke~\cite{hr06} devised an $O(\sqrt{n})$-approximation for directed sparsest cut via a new LP-rounding technique for the fractional relaxation, and Agarwal, Alon, and Charikar~\cite{aac07} later improved this to $\widetilde{O}(n^{11/23})$ with a refined randomized rounding scheme. Yoshida~\cite{yoshida16} introduced a nonlinear Laplacian for digraphs and proved a Cheeger-type inequality. For directed edge expansion, Lau, Tung, and Wang~\cite{Lau2024} combined triangle inequalities with the reweighted eigenvalue formulation to obtain an almost-linear-time $O(\sqrt{\log n})$-approximation via a strengthened SDP relaxation.

\section{Preliminaries}
\paragraph{Notation.} We use boldface uppercase and lowercase letters to denote matrices and vectors, respectively. For a vector $\Bx \in \R^n$, let $x_i$ or $x(i)$ denote the $i$-th entry of $\Bx$; for $S \subset [n]$, let $x(S) := \sum_{i \in S} x(i)$. Given two vectors $\Bx, \By \in \R^n$, let $\inprod{\Bx, \By} := \Bx^\top \By$. 
For any matrix $\BA \in \R^{n \times n}$, let $\tr(\BA)$ denote the trace of $\BA$. Given two matrices $\BA, \BB \in \R^{n \times n}$, we define $\inprod{\BA, \BB} := \tr(\BA \BB)$.
The operator and Frobenius norms are denoted by $\norm{\BA}$ and $\norm{\BA}_F$, respectively.

Given an undirected graph $G = (V, E)$ and two disjoint subsets $S, T \subset V$, let $E(S, T)$ denote the set of edges with one endpoint in $S$ and the other in $T$. For an edge $e = (u, v) \in E$, let $w(e)$ or $w(u, v)$ denote its weight. For an edge subset $E' \subset E$, let $w(E') := \sum_{e \in E'} w(e)$. For a vertex $v \in V$, let $\deg(v)$ denote the degree of $v$, i.e., $\deg(v) = \sum_{\text{$e \in E$: incident to $v$}} w(e)$. For a vertex subset $S \subset V$, let $E(S)$ denote the set of edges with both endpoints in $S$, and $\vol(S) := \sum_{v \in S} \deg(v)$. For graph $G$, let $\BL \in \R^{V \times V}$ denote its Laplacian matrix, and diagonal matrix $\BD \in \R^{E \times E}$ be the degree matrix such that the $i$-th diagonal entry is the degree of $i \in V$. Furthermore, let $\Lap := \BD^{-1/2} \BL \BD^{-1/2} \in \R^{V \times V}$ denote the normalized Laplacian matrix of $G$, and $0 = \lambda_1 \le \cdots \le \lambda_n \le 2$ be the eigenvalues of $\Lap$.

\begin{theorem}[Flow decomposition theorem] \label{thm:flow_dec}
Let $G = (V, E)$ be an undirected graph, $s, t \in V$ be sink and source vertices, and $c: E \to \R_{+}$ be an edge capacity function. 
For an $s$--$t$ feasible flow $f$, there is a collection of positive values $f_1, \cdots, f_k \geq 0$ and a collection of $s$--$t$ paths $P_1, \cdots, P_k$ such that 
\begin{itemize}
    \item $k \le |E|$; 
    \item the flow $f$ sends $f_i$ unit of flow through $P_i$ for each $i \in [k]$.
\end{itemize}
Furthermore, if $c$ is integer-valued, then $f_1, \cdots, f_k$ can be taken to be integers.
In other words, flow $f$ can be decomposed as a multiset of $s$--$t$ paths.
\end{theorem}
There is a randomized algorithm that finds a flow decomposition in $O(m \log n)$ time~\cite{Lee2013}.

\begin{definition}[demand graph]\label{def:demand_g}
    For a multiset of paths $\caP$ in an undirected graph $G = (V, E)$, we define the \emph{demand graph} $M$ of $\caP$ as the following multigraph.
    The vertex set of $M$ is $V$.
    For each $i, j \in V$, $M$ has $p_{i, j}$ many parallel edges $(i, j)$, where $p_{i, j}$ is the number of paths in $\caP$ between $i$ and $j$.

    For a demand graph $M = (V, E)$ with $\abs{V} = n$, define $\BD_M$ be an $n \times n$ diagonal matrix such that $(\BD_M)_{i,i}$ equals the degree of vertex $i$ in $M$. 
    Also define $\BA_M$ be an $n \times n$ matrix such that $(\BA_M)_{i,j}$ equals to the number of edges between vertices $i$ and $j$ in $M$.
\end{definition}

\subsection{Concentration Inequalities}
We employ the following standard concentration bounds for Gaussian random variables.

\begin{lemma}\label{lem:gaussian-concentration}
Let $\Bv \in \R^n$ be a vector, $\Bg \sim N(\mb{0}, \BI_n)$ be a standard Gaussian random variable, and $X = \inprod{\Bg, \Bv}$.
Then $\E[X] = 0$ and $\E[X^2] = \norm{\Bv}_2^2$.
    Furthermore, for any $t > 0$, 
    \begin{align*}
        \Pr\left(\abs{X} > t \cdot \norm{\Bv}_2 \right) \leq 2 \cdot \exp(-t^2/2). 
    \end{align*} 
\end{lemma}

\begin{lemma}[{Laurent-Massart bound, \cite[Lemma~1]{Laurent2000}}]\label{lem:Laurent-Massart}
    For i.i.d.~standard Gaussian random variables $g_1, \dots, g_n \in \R$ and scalars $a_1, \dots, a_n \geq 0$, we have that for any $t > 0$, 
    \begin{align*}
        \Pr\left(\sum_{i=1}^n a_i(g_i^2 - 1) \leq -2 \sqrt{t} \cdot \norm{\Ba}_2 \right) \leq \exp(-t),
    \end{align*} 
where $\Ba = [a_1, \dots, a_n]^\top$.
\end{lemma}

\subsection{A Useful Result on Skew-Symmetric Graphs}
An \emph{involution} on a set $V$ is a mapping $\sigma: V \to V$ such that $\sigma(v) \neq v$ and $\sigma(\sigma(v)) = v$ for all $v \in V$.
A directed graph $G = (V, E)$ is said to be \emph{skew-symmetric} if there exists an involution $\sigma: V \to V$ on the vertex set $V$ such that $(u, v) \in E$ implies $(\sigma(v), \sigma(u)) \in E$.
The edge weight $w: E \to \R$ is said to be \emph{skew-symmetric} (w.r.t.~involution $\sigma$) if $w(u, v) = w(\sigma(v), \sigma(u))$ for all $(u, v) \in E$.
We can also define a skew-symmetric undirected graph with the same condition.

\begin{lemma}\label{lem:ssym-cut}
    Let $G = (V, E)$ be a skew-symmetric directed graph and $w$ be a nonnegative skew-symmetric edge weight w.r.t.~the same involution $\sigma$.
    Let $X \subseteq V$ be a vertex subset and $X'$ be the subset of $X$ obtained by dropping all $v \in X$ with $\sigma(v) \in X$.
    Then, $w(E(X, \ol{X})) \geq w(E(X', \ol{X'}))$.
    The same holds for the undirected case.
\end{lemma}
\begin{proof}
    By the skew-symmetry of $G$ and $w$, we have 
    \begin{align*}
        w(E(X, \ol{X})) &= \frac{1}{2} \sum_{(u, v) \in E} w(u, v) \left(\kappa_{(u,v)}(X) + \kappa_{(\sigma(v), \sigma(u))}(X)\right), 
    \end{align*}
    where $\kappa_{(u,v)}(X)$ denotes the indicator function that takes $1$ if $u \in X$ and $v \notin X$, and $0$ otherwise.
    Since $w$ is nonnegative, it suffices to show that 
    \begin{align*}
        \kappa_{(u,v)}(X) + \kappa_{(\sigma(v), \sigma(u))}(X) \geq \kappa_{(u,v)}(X') + \kappa_{(\sigma(v), \sigma(u))}(X')
    \end{align*}
    for each $(u, v) \in E$.

    First, consider the case where $X$ contains both $u$ and $\sigma(u)$. 
    Then, $u, \sigma(u) \notin X'$.
    We have the following four cases:
    \begin{enumerate}[(1)]
        \item $v, \sigma(v) \in X$: Then, no edges are cut by both $X$ and $X'$. So the both sides are zero.
        
        \item $v \in X$ and $\sigma(v) \notin X$: Again, no edges are cut by both $X$ and $X'$. 

        \item $v \notin X$ and $\sigma(v) \in X$: Then, $X$ cuts only $(u,v)$ and $X'$ cuts only $(\sigma(v), \sigma(u))$. So the both sides are one.
       
        \item $v, \sigma(v) \notin X$: Then $X$ cuts only $(u,v)$, whereas $X'$ cuts none of them. So the left-hand side is one and the right-hand side is zero.
        \end{enumerate}
    This completes the proof for this case.

    Next, consider the case that $X$ contains exactly one of $u$ and $\sigma(u)$.
    By symmetry, we can further assume that $u \in X$ and $\sigma(u) \notin X$.
    If $v, \sigma(v) \in X$, then $X$ cuts only $(\sigma(v), \sigma(u))$ and $X'$ cuts only $(u, v)$. So the both sides are one.
    In the other cases, the cut edges remain the same in $X$ and $X'$. 

    Finally, suppose that $X$ contains neither $u$ nor $\sigma(u)$. 
    If $v, \sigma(v) \in X$, then $X$ cuts only $(\sigma(v), \sigma(u))$ and $X'$ cuts none.
    So the left-hand side is one, and the right-hand side is zero. 
    In the other cases, the cut edges remain the same in $X$ and $X'$. 
    This completes the proof for the directed case.

    The undirected case reduces to the directed case by bidirecting each edge.
\end{proof}

\section{Cut-Matching Game for Undirected Bipartiteness Ratio}
In this section, we generalize the cut-matching game framework for approximating the undirected bipartiteness ratio.

\subsection{A Flow-Cut Characterization of Bipartiteness Ratio}\label{sec:bipartite-flow-cut}
We will introduce a convenient representation of the bipartiteness ratio with cuts in an auxiliary graph.

\begin{definition}[auxiliary graph $G'$]
    The auxiliary graph $G' = (V', E')$ is defined as follows:
    \begin{itemize}
        \item $V' = V^+ \cup V^-$, where $V^+$ and $V^-$ are disjoint copies of $V$.
        \item $E' = \bigcup_{(i, j) \in E}\{(i^+, j^-) , (i^-, j^+)\}$, where $i^+$ and $i^-$ denote the copies of vertex $i$ in $V^+$ and $V^-$, respectively.  
    \end{itemize}
For $X \subseteq V$, denote by $X^+$ (resp.~$X^-$) the corresponding set of $X$ in $V^+$ (resp.~$V^-$). 
Given a vertex weight $b$ on $V$, let $b'$ be the weight on $V'$ such that $b'(u) = b(i)$ for $i \in V$ and $u \in \{i^+, i^-\}$.
With a slight abuse of notation, we also denote the induced vertex weights on $V'$ by $b$, and similarly, the induced edge weights on $E'$ by $w$. 
\end{definition}

See \Cref{fig:ori_bip} for an illustration of the auxiliary graph $G'$.
Note that $G$ and $w$ are skew-symmetric with involution $\sigma(i^+) = i^-$ and $\sigma(i^-) = i^+$ for $i \in V$.
Recall that each non-zero $\{0,\pm 1\}$-vector $\Bx$ corresponds to a partition of $V = L \cup R \cup Z$ such that $L = \{i \in V \mid x_i = 1\}$, $R = \{j \in V \mid x_j = -1\}$, and $Z = \{k \in V \mid x_k = 0\}$, respectively. 

\begin{claim} \label{cla:betax_conds}
For any $\Bx \in \{-1, 0, 1\}^V \setminus \{\mb{0}\}$ with partition $L \cup R \cup Z$ of $V$, let $S := L^+ \cup R^-$ and $\ol{S} := V' \setminus S$, then we have  
    \[
        \beta_b(\Bx) = \frac{w(E'(S, \ol{S}))}{b(S)},  
    \]
    and thus 
    \begin{align}\label{eq:beta-G'}
        \beta_b(G) = \min_{S = L^+ \cup R^-,\ \textnormal{disjoint } L, R \subseteq V} \frac{w(E'(S, \ol{S}))}{b(S)}.
    \end{align}
\end{claim} 

\begin{proof}
    Since $w(i^+, j^-) = w(i^-, j^+) = w(e)$ for each $e = (i, j) \in E$, we have
    \begin{align*}
        w(E'(S, \ol{S})) &= \sum_{e = (i, j) \in E} w(e) (\kappa_{(i^+,j^-)}(S) + \kappa_{(i^-,j^+)}(S)),
    \end{align*}
    where $\kappa_{u,v}(S)$ denotes the indicator function that takes $1$ if $(u, v)$ is cut by $S$ and $0$ otherwise.
    One can check that $|x_i + x_j| = \kappa_{(i^+,j^-)}(S) + \kappa_{(i^-,j^+)}(S)$ for any symmetric $S = L^+ \cup R^-$ and the corresponding $\Bx \in \{0, \pm 1\}^V$.
    Therefore, $w(E'(S, \ol{S})) = \sum_{e = (i, j) \in E} w(e) \cdot |x_i + x_j|$.

    Additionally, the denominator of $\beta_b(\Bx)$ satisfies that 
\begin{align} \label{eq:betabx_den}
\sum_{i \in V} b(i) \cdot |x_i| = \sum_{i \in L \cup R} b(i) = \sum_{u \in L^+ \cup R^-} b(u) = b(S),  
\end{align} 
where the last step follows from $b(i) = b(i^+) = b(i^-)$ for $i \in V$. 
This completes the proof.
\end{proof}

We say that a pair $(A, B)$ of subsets of $V'$ is \emph{symmetric} if there exist disjoint $L, R \subseteq V$ such that
\begin{align}\label{eq:LR-ST}
A = L^+ \cup R^-,\ B = L^- \cup R^+.
\end{align}
Note that if $(A, B)$ is symmetric, then $b(A) = b(B) = b(L) + b(R)$.
We now introduce the concept of \emph{well-linkedness} which will be used to characterize the bipartiteness ratio.
Let $r > 0$ be a parameter and $(A, B)$ be a symmetric pair of subsets of $V'$.
Consider the following undirected auxiliary network $N_{A,B,r}$ (see \cref{fig:N_STr}). 
The vertex set of the network is $V' \cup \{s^+, s^-\}$, where $s^+$ and $s^-$ are the super source and super sink, respectively. Connect $s^+$ to each $u \in A$ with an edge of capacity $b(u)$. Similarly, connect each $v \in B$ to $s^-$ with an edge of capacity $b(v)$.
Finally, connect the same edges as in $G'$ between $V^+$ and $V^-$ such that the capacity of each edge $e \in E'$ is set to be $w(e) / r$.
By construction, $N_{A,B,r}$ and its capacity function are skew-symmetric, where we regard the capacity function as an edge weight.

\begin{figure}[t] 
    \centering
    \includegraphics[scale = 0.8]{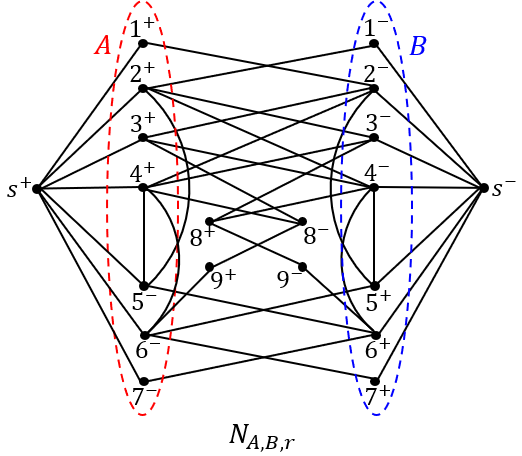}
    \caption{The figure of auxiliary network $N_{A, B, r}$, where $A = L^+ \cup R^-$, $B = L^- \cup R^+$, and $s^+, s^-$ are the source and sink. Additionally, $s^+$ has an edge to each vertex $u \in A$ with capacity $b(u)$, each $v \in B$ has an edge to $s^-$ with capacity $b(v)$, and each edge $e \in E'$ has capacity $w(e) / r$.} 
    \label{fig:N_STr}
\end{figure}

\begin{definition}[well-linkedness in $G'$]\label{def:well-linked}
    A feasible $s^+$--$s^-$ flow of the auxiliary network $N_{A,B,r}$ is said to be \emph{saturating} if all the edges from $s^+$ and to $s^-$ are saturated, i.e., their capacities are attained by the flow.
    We say that a symmetric pair $(A, B)$ is \emph{$r$-well-linked} if there exists an $s^+$--$s^-$ saturating flow in the auxiliary network $N_{A,B,r}$. We say that $G'$ is $r$-well-linked if any symmetric $(A, B)$ is $r$-well-linked.
\end{definition}

We will show that the $r$-well-linkedness of $G'$ is equivalent to the bipartiteness ratio $\beta_b(G)$ being at least $r$.
To begin, we introduce the following lemma. Let $X \subseteq V' \cup \{s^+, s^-\}$. A vertex $i \in V \cup \{s\}$ is \emph{inconsistent} in $X$ if both $i^+$ and $i^-$ are in $X$. We say that $X$ is \emph{consistent} if $X$ does not contain any inconsistent vertices.
The following lemma is immediate from the skew-symmetry of $N_{A, B, r}$ and \Cref{lem:ssym-cut}.

\begin{lemma}\label{lem:cut}
    Let $(A, B)$ be a symmetric pair and $X$ be a minimum $s^+$--$s^-$ cut in $N_{A,B,r}$.
    Let $X'$ be the consistent set obtained by dropping the copies of all inconsistent vertices from $X$.
    Then, $X'$ is also a minimum $s^+$--$s^-$ cut.
\end{lemma}

Now we are ready to show the following theorem.

\begin{theorem}[Well-linkedness characterization of $b$-bipartiteness ratio]\label{thm:linked}
    $\beta_b(G) \geq r$ if and only if $G'$ is $r$-well-linked.
\end{theorem}
\begin{proof}
    By the max-flow min-cut theorem, it suffices to show that $\beta_b(G) \geq r$ if and only if the minimum $s^+$--$s^-$ cut of $N_{A,B,r}$ is at least $b(A)$ for any symmetric $(A, B)$.

    (If part)
    Take an arbitrary symmetric $(A, B)$ and assume that the minimum cut in $N_{A,B,r}$ is at least $b(A)$.
    Because $A \cup s^+$ is an $s^+$--$s^-$ cut of value $r^{-1} \cdot w(E'(A, \ol{A}))$, we have $r^{-1} \cdot w(E'(A, \ol{A})) \geq b(A)$.
    Since $(A, B)$ is arbitrary, this implies that $\beta_b(G) \geq r$.
    
    (Only if part)
    We show the contrapositive.
    Assume that $G'$ is not $r$-well-linked, i.e., for some symmetric $(A, B)$, the minimum $s^+$--$s^-$ cut in $N_{A,B,r}$ is less than $b(A)$.
    Let $X$ be a minimum $s^+$--$s^-$ cut in $N_{A,B,r}$.
    By \cref{lem:cut}, without loss of generality, we can assume that $X$ is consistent, i.e., there exists consistent $S \subseteq V'$ such that $X = S \cup s^+$.
    Then, the cut value of $X$ can be bounded as 
    \begin{align*}
        b(A) - b(A \cap S) + b(B \cap S) + r^{-1} \cdot w(E'(S, \ol{S})) < b(A),
    \end{align*}
    and therefore
    \begin{align*}
        r^{-1} \cdot w(E'(S, \ol{S})) < b(A \cap S) - b(B \cap S) \leq b(S).
    \end{align*}
    Thus $S \neq \emptyset$ (otherwise $b(S) = 0$ and the above inequality is violated) and we have $\beta_b(S) < r$.
\end{proof}

Lastly, we rephrase \Cref{thm:linked} in terms of congestion, which is convenient for dealing with different values of $r$.
Let $N_{A,B}$ denote the auxiliary network $N_{A,B,r}$ with $r = 1$.
Recall that the congestion of a flow $f$ in $N_{A,B}$ (not necessarily satisfying the capacity constraint) is the maximum ratio of $f(e)/c(e)$ over edge $e$ in $N_{A,B}$, where $c$ is the edge capacity function.

\begin{corollary}[Congestion characterization of $b$-bipartiteness ratio]\label{cor:congestion}
    $\beta_b(G) \geq r$ if and only if for any symmetric $(A, B)$, the auxiliary network $N_{A,B}$ has an $s^+$--$s^-$ saturating flow with congestion at most $1/r$.
\end{corollary}

\subsection{Cut-Matching Game via Matrix Multiplicative Weight Update}
Following \cite{Arora2016,Lau2024}, we propose a cut-matching game (see \cref{alg:cut-matching}) derived from MMWU. 

\begin{algorithm}[t] 
    \caption{Cut-Matching Game for Bipartiteness Ratio}
    \label{alg:cut-matching}   

    \Input{an undirected graph $G = (V, E)$ and $r > 0$ with $1/r$ being an integer.}

    \Output{either a vector $\bx \in \{0, \pm 1\}^V$ with $\beta_b(\bx) < r$ or a certificate $H$ that proves $\beta_b(G) \geq \Omega(r/\log n)$.}

    Set $T = O(\log^2 n)$, $\Tproj = O(\log n)$, $\delta \in (0, 1)$.  

    \For{$t = 1$ \KwTo $T$}
    {
       Compute an approximate Gram decomposition $\Bv_1, \cdots, \Bv_n$ of $\BD_b^{-1/2} \BX_t \BD_b^{-1/2}$ by \Cref{lem:app_gram}, where $\BX_t$ is the MMWU iterate \eqref{eq:MMWU}.

       Sample a standard gaussian vector $\Bg \sim \mc{N}(\mb{0}, \BI_n)$ and compute $\wt{v}_i \gets \Bg^\top \Bv_i$ for each $i \in V$.
       If $\sum_{i \in V} b(i) \abs{\wt{v}_i}^2 < 1/4$, then sample $\Bg$ again and repeat.
       If it fails to find such $\Bg$ after $\Tproj$ times, then \textbf{fail}.

       Let $L' \gets \{i \in V \mid \wt{v}_i > 0\}$ and $R' \gets \{i \in V \mid \wt{v}_i < 0\}$.

       \uIf{$\sum_{i \in L'} b(i)\abs{\wt{v}_i}^2 < \sum_{i \in R'} b(i)\abs{\wt{v}_i}^2$}
       {
          $(L', R') \gets (R', L')$.
       }

       Let $(L, R) \gets (L', \emptyset)$ and $(A, B)$ be the corresponding symmetric sink-source pair \eqref{eq:LR-ST}.

    \uIf{$(A, B)$ is not $r$-well-linked} 
    {
        Find a consistent minimum $s^+$--$s^-$ cut in the auxiliary network $N_{A,B,r}$ and $\bx \in \{0, \pm 1\}^V$ be the vector corresponding to the cut.

        \Return $\bx$ \Comment{This case we find $\bx$ such that $\beta_b(\bx) < r$.} \label{line:return-x}

    }
    \Else
    {
    Find an $s^+$--$s^-$ integral saturating flow in the auxiliary network $N_{A,B,r}$.

    Decompose the flow into a multiset $\caP_t$ of odd $L$--$L$ paths.

    Let $M_t$ be the demand graph of $\caP_t$.

    $\BF_t \gets \BD_b^{-1/2} \sum_{(i, j) \in M_t} (\Be_i + \Be_j) (\Be_i + \Be_j)^\top \BD_b^{-1/2}$.
    }
    }    
    \Return $H \gets M_1 \oplus \cdots \oplus M_T$ \Comment{This case we find a certificate $H$ that proves $\beta_b(G) \geq \Omega(r/\log n)$.} \label{line:return-H}
\end{algorithm}

In MMWU, we maintain $n \times n$ symmetric positive definite matrix $\BX_t$ with trace one (i.e., density matrix).
For each round $t$, we receive a symmetric matrix $\BF_t$ and update the density matrix by 
\begin{align}\label{eq:MMWU}
    \BX_{t+1} = \frac{\exp(-\delta \sum_{\tau=1}^t \BF_\tau)}{\tr(\exp(-\delta \sum_{\tau=1}^t \BF_\tau))},
\end{align}
where $\delta > 0$ is a parameter called step size.
Conventionally, we define $\BX_1 = \frac{1}{n} \BI_n$.
The following is the standard regret bound of MMWU. 

\begin{lemma}[Theorem 10 in \cite{kale07}] \label{lem:mmwu}
    Given $\mb{0} \preceq \BF_t \preceq \rho \BI$ and $\delta \in (0, 1)$, it holds that 
    \[
    \lambda_{\min}\left(\sum_{t =1}^{T} \BF_t\right) \ge (1 - \rho \delta) \cdot \sum_{t=1}^{T} \inprod{\BF_t, \BX_t} - \frac{\ln n}{\delta}. 
    \]
\end{lemma}

To connect MMWU with a cut-matching game, we need an ``oracle'' that finds either (i) an $n \times n$ symmetric matrix $\BF_t$ with $\inprod{\BF_t, \BX_t} \geq \gamma$ and $\BF_t \preceq \rho \BI$, or (ii) symmetric $(A, B)$ that is not $r$-well-linked, where $r$ is a guess of $\beta_b(G)$.
The parameter $\rho$ is often called the \emph{width}.
Furthermore, $\BF_t$ must correspond to a subgraph of $G$ in some manner so that if $\lambda_{\min}\left(\sum_{t=1}^T \BF_t\right)$ is large, then the bipartiteness ratio of $G$ is also large. 
Later, we take $\BF_t$ as the demand matrix of a subgraph $G_t$ of $G$. 

Suppose that the oracle returns $(A, B)$ that is $r$-well-linked.
Then, we can find a saturating  $s^+$--$s^-$ flow of in $N_{A,B,r}$ by definition.
By the flow decomposition theorem (see \cref{thm:flow_dec}), there is a multiset of $A$--$B$ paths realizing the flow.
These paths correspond to a multiset $\caP$ of odd $L$--$L$ paths, odd $R$--$R$ paths, and even $L$--$R$ paths in $G$.
Let $M$ be the demand graph of $\caP$ (see \cref{def:demand_g}). 
Note that $\deg_M(i) = 2b(i)$ for each $i \in L \cup R$ because the flow saturates all edges connecting to $s^+$ and $s^-$. 
We set 
\begin{align*}
    \BF_t = \BD_b^{-1/2}\sum_{(i, j) \in M_t} (\be_i + \be_j) (\be_i + \be_j)^\top \BD_b^{-1/2} = \BD_b^{-1/2}(\BD_{M_t} + \BA_{M_t}) \BD_b^{-1/2}.
\end{align*}
Note that $i$ and $j$ can be identical if $(i, j)$ is a self-loop.
Since $\deg_{M_t}(i) \leq 2b(i)$ for $i \in V$, we have $\BD_{M_t} \preceq 2 \BD_b$ and 

\begin{align*}
\BF_t = \BD_b^{-1/2} (\BD_{M_t} + \BA_{M_t}) \BD_b^{-1/2} \preceq 2 \BD_b^{-1/2} \BD_{M_t} \BD_b^{-1/2} \preceq 4 \BI_n, 
\end{align*} 
where the second step follows from $\BA_{M_t} \preceq \BD_{M_t}$; the third step follows from $\BD_{M_t} \preceq 2 \BD_b$. Therefore, we have $\BF_t \preceq 4 \BI$ and $\rho = O(1)$.

The remaining task is to bound $\gamma$. Letting $\Bv_1, \dots, \Bv_n$ be a Gram decomposition of $\BD_b^{-1/2} \BX_t \BD_b^{-1/2}$, that is, $\BD_b^{-1/2} \BX_t \BD_b^{-1/2} = \BV^\top \BV$ with $\BV = [\Bv_1, \cdots, \Bv_n]$, then we have  
\begin{align*}
    \inprod{\BF_t, \BX_t} 
    &= \left<\BD_b^{-1/2} \sum_{(i, j) \in M_t} (\Be_i + \Be_j) (\Be_i + \Be_j)^\top \BD_b^{-1/2}, \BX_t \right> \\ 
    &= \sum_{(i, j) \in M_t} \inprod{\BD_b^{-1/2} (\Be_i + \Be_j) (\Be_i + \Be_j)^\top \BD_b^{-1/2}, \BX_t} \\ 
    &= \sum_{(i, j) \in M_t} \inprod{(\be_i + \be_j)(\be_i + \be_j)^\top, \BD_b^{-1/2} \BX_t \BD_b^{-1/2}} \\
    &= \sum_{(i, j) \in M_t} \inprod{(\be_i + \be_j)(\be_i + \be_j)^\top, \BV^\top \BV} \\ 
    &= \sum_{(i, j) \in M_t} (\BV (\Be_i + \Be_j))^\top \BV (\Be_i + \Be_j) \\
    &= \sum_{(i, j) \in M_t} \norm{\Bv_i + \Bv_j}^2. 
\end{align*}
Note that 
\begin{align} \label{eq:sum_bi_vi}
\sum_{i \in V} b(i) \cdot \norm{\Bv_i}^2 =\tr(\BV \BD_b \BV^\top) = \tr(\BD_b^{1/2} \BV^\top \BV \BD_b^{1/2}) = \tr(\BX_t) = 1.  
\end{align} 
So now the goal is to find $(L, R)$ such that $\sum_{(i, j) \in M} \norm{\Bv_i + \Bv_j}^2$ is large for \emph{any} demand graph $M_t$.
We will show that there is a simple way of choosing $(L, R)$ with $\gamma = \Omega(1/\log n)$.
In what follows, we denote $M_t$ by simply $M$.

\subsubsection{One-dimensional case}
First, let us pretend that the vector $\Bv_i$ is a scalar $v_i$.
The general case will be reduced to this case by the standard Gaussian projection trick. Let $L' = \{i \in V : v_i > 0\}$ and $R' = \{i \in V : v_i < 0\}$, then \eqref{eq:sum_bi_vi} gives $\sum_{i \in L'} b(i) \cdot \abs{v_i}^2 + \sum_{i \in R'} b(i) \cdot \abs{v_i}^2 = 1$.
Without loss of generality, we can assume that $\sum_{i \in L'} b(i) \cdot \abs{v_i}^2 \geq 1/2$; otherwise just swap $L'$ and $R'$. Let $(L, R) = (L', \emptyset)$, then $M$ consists of edges (possibly self-loops) connecting vertices in $L$.
Furthermore, $\deg_M(i) = 2b(i)$ for each $i \in L$ by construction.
For each $\{i, j\} \in M$, 
\begin{align*}
    \abs{v_i + v_j}^2 = (\abs{v_i} + \abs{v_j})^2 = \abs{v_i}^2 + \abs{v_j}^2 + 2\abs{v_i}\abs{v_j} \geq \abs{v_i}^2 + \abs{v_j}^2.
\end{align*}
Summing this over $\{i, j\} \in M$, we have
\begin{align*}
    \sum_{\{i, j\} \in M} \abs{v_i + v_j}^2  
    \geq \sum_{\{i, j\} \in M} (\abs{v_i}^2 + \abs{v_j}^2)
    = \sum_{i \in L} \deg_M(i) \cdot \abs{v_i}^2
    = 2\sum_{i \in L} b(i) \cdot \abs{v_i}^2 
    \geq 1, 
\end{align*}
where the last step follows from $L = L'$ and $\sum_{i \in L'} b(i) \cdot |v_i|^2 \ge 1/2$. Therefore, in this case, $\gamma = 1$.

\subsubsection{General case}

Let us now describe how to choose $(L, R)$ formally. We first compute a Gram decomposition $\Bv_1, \dots, \Bv_n$ of $\BD_b^{1/2} \BX_t \BD_b^{1/2}$, then sample $\Bg \sim N(\mb{0}, \BI_n)$ and compute $\wt{v_i} = \inprod{\Bg, \Bv_i} \in \R$ for each $i \in V$. By \Cref{lem:gaussian-concentration}, we have $\E_g[\wt v_i^2] = \norm{\Bv_i}_2^2$ and 
\begin{align}\label{eq:tilde-v-ineq}
     \abs{\wt v_i + \wt v_j} = \abs{\inprod{\Bv_i + \Bv_j, \Bg}} \leq O\left(\sqrt{\log n}\right) \cdot \norm{\Bv_i + \Bv_j}_2 
\end{align}
for any $i, j$ with probability $1 - 1/\poly(n)$.

We also need to ensure that $\sum_{i \in V} b(i) \cdot \abs{\wt v_i}^2$ is not too small.
Note that $\E\left[\sum_{i \in V} b(i) \cdot \abs{\wt v_i}^2 \right] = \sum_{i \in V} b(i) \cdot \norm{\Bv_i}_2^2 = 1$. Next we will show that $\sum_{i \in V} b(i) \cdot \abs{\wt v_i}^2 \geq 1/2$ with at least a constant probability.
\begin{lemma}
    \[
        \Pr\left(\sum_{i \in V} b(i) \cdot \abs{\tilde v_i}^2 < \frac{1}{2} \right) \leq e^{-1/16}.
    \]
\end{lemma}
\begin{proof}
    Note that
    \begin{align*}
        \sum_{i \in V} b(i) \cdot \abs{\wt v_i}^2 
        = \sum_{i \in V} b(i) \cdot \inprod{\Bg, \Bv_i}^2 
        = \Bg^\top \left(\sum_{i \in V} b(i) \cdot \Bv_i \Bv_i^\top \right) \Bg, 
    \end{align*}
where the matrix $\sum_{i \in V} b(i) \Bv_i \Bv_i^\top$ is positive semidefinite and has trace one.
    Since the gaussian distribution is invariant under orthogonal transformations, we can assume that $\sum_{i \in V} b(i) \Bv_i \Bv_i^\top = \Diag(\lambda_1, \dots, \lambda_n)$, where $\lambda_1, \dots, \lambda_n \geq 0$ are the eigenvalues of $\sum_{i \in V} b(i) \cdot \Bv_i \Bv_i^\top$ and hence $\sum_i \lambda_i = 1$. As a consequence, we have  
    \begin{align*}
        \sum_{i \in V} b(i) \cdot \abs{\wt v_i}^2 = \Bg^\top \Diag(\lambda_1, \cdots, \lambda_n) \Bg  
        = \sum_{i=1}^n \lambda_i g_i^2,
    \end{align*}
    where $g_1, \dots, g_n \in \R$ are independent standard Gaussian random variables.
    By the Laurent-Massart bound (see \Cref{lem:Laurent-Massart}), we have that for any $t > 0$, 
    \[
        \Pr\left(\sum_{i=1}^n \lambda_i g_i^2 < 1 - 2 \sqrt{t} \cdot \norm{\bm{\lambda}}_2 \right) \leq \exp(-t), 
    \]
    where $\bm{\lambda} = [\lambda_1, \cdots, \lambda_n]^\top$. 
    Setting $t = \frac{1}{16\norm{\bm{\lambda}}_2^2}$ gives rise to 
    \[
        \Pr\left(\sum_{i=1}^n \lambda_i g_i^2 < \frac{1}{2} \right) \leq \exp\left(-\frac{1}{16\norm{\bm{\lambda}}_2^2}\right) \leq \exp\left(-\frac{1}{16 \norm{\bm{\lambda}}_1^2}\right) = e^{-1/16},
    \]
    where the last step follows from $\norm{\bm{\lambda}}_1 = \sum_i \lambda_i = 1$.
\end{proof}

Therefore, we can repeatedly sample $\Bg$ at most $O(\log n)$ times until 
\begin{align}\label{eq:sum-tilde-v}
\sum_{i \in V} b(i) \cdot \abs{\tilde v_i}^2 \geq \frac{1}{2}
\end{align}
holds.
The success probability is at least $1 - 1/\poly(n)$ by the above lemma.
Assume that \eqref{eq:tilde-v-ineq} and \eqref{eq:sum-tilde-v} hold in what follows.
Let $L' = \{i \in V: \tilde v_i > 0\}$ and $R' = \{i \in V: \tilde v_i < 0\}$.
Again, without loss of generality, we can assume that $\sum_{i \in L'} b(i) \cdot \abs{\tilde v_i}^2 \geq 1/4$.
Setting $(L, R) = (L', \emptyset)$ gives us 
\begin{align*}
    \sum_{\{i, j\} \in M} \norm{\Bv_i + \Bv_j}^2
    \geq \frac{1}{O(\log n)} \sum_{\{i, j\} \in M} \abs{\wt{v_i} + \wt{v_j}}^2
    \geq \frac{1}{O(\log n)} \sum_{i \in L} b(i) \cdot \abs{\wt{v_i}}^2
    \geq \frac{1}{O(\log n)}.
\end{align*}
Therefore, we have $\gamma = \Omega(1/\log n)$.

\subsubsection{Putting things together}

Assume that the game goes for $T$ rounds. We have matrices $\BF_t$ for each $t \in [T]$ which correspond to a multiset $\caP_t$ of paths with congestion at most $r^{-1}$ in $G'$.
Substituting $\rho = O(1)$, $\delta = O(1)$, and $\gamma = \Omega(1/\log n)$ to the MMWU bound (\cref{lem:mmwu}), we have
\begin{align*}
\lambda_{\min}\left(\sum_{t=1}^T \BF_t\right) 
    &\gtrsim \frac{T}{\log n } - \log n.
\end{align*}
Therefore, setting $T = O(\log^2 n)$, we have $\lambda_{\min}(\sum_{t=1}^T \BF_t) = \Omega(T/\log n) = \Omega(\log n)$.
Recall that each $\BF_t$ corresponds to the demand graph $M_t$ of $\caP_t$.
Let $H := M_1 \oplus \dots \oplus M_T$ be the multigraph on $V$.
From the eigenvalue lower bound, we have the following lemma.

\begin{lemma}
    $\beta_b(H) = \Omega(\log n)$. 
\end{lemma}

\begin{proof}
By the above argument, we have $\lambda_{\min}(\sum_{t=1}^T \BF_t) = \Omega(\log n)$.
Therefore, it suffices to show that $\beta_b(H) \gtrsim \lambda_{\min}(\sum_{t=1}^T \BF_t)$.
Since the minimum eigenvalue of a symmetric matrix equals the minimum of its Rayleigh quotient, we have
\begin{align*}
    \lambda_{\min}\left(\sum_{t=1}^T \BF_t \right) 
    &= \min_{\Bx \in \R^n \setminus \{\zeros\}} \frac{\Bx^\top \sum_{t=1}^T \BF_t \Bx}{\Bx^\top \Bx} \\
    &= \min_{\Bx \in \R^n \setminus \{\zeros\}} \frac{\Bx^\top \BD_b^{-1/2} \sum_{t=1}^T (\BD_{M_t} + \BA_{M_t}) \BD_b^{-1/2} \Bx}{\Bx^\top \Bx} \\
    &= \min_{\Bx \in \R^n \setminus \{\zeros\}} \frac{\Bx^\top \sum_{t=1}^T (\BD_{M_t} + \BA_{M_t}) \Bx}{\Bx^\top \BD_b \Bx}.
\end{align*}
For any $\{0, \pm 1\}^n$-vector $\Bx$, we have
\begin{align*}
    \Bx^\top \sum_{t=1}^T (\BD_{M_t} + \BA_{M_t}) \Bx &= \sum_{(i,j) \in E(H)} w(i, j) \cdot (x_i + x_j)^2 
    \leq 2 \sum_{(i,j) \in E(H)} w(i, j) \cdot \abs{x_i + x_j}, \\
    \Bx^\top \BD_b \Bx &= \sum_{i \in V} b(i) \cdot x_i^2 = \sum_{i \in V} b(i) \cdot  \abs{x_i}. 
\end{align*}
Therefore,
\begin{align*}
    \beta_b(H) &= \min_{\Bx \in \{0, \pm 1\}^n \setminus \{\zeros\}} \frac{\sum_{(i,j) \in E(H)} w(i, j) \cdot \abs{x_i + x_j}}{\sum_{i \in V} b(i) \cdot \abs{x_i}} 
    \geq \frac{1}{2} \cdot \min_{\Bx \in \{0, \pm 1\}^n \setminus \{\zeros\}} \frac{\Bx^\top \sum_{t=1}^T (\BD_{M_t} + \BA_{M_t}) \Bx}{\Bx^\top \BD_b \Bx} \\
    &\geq \frac{1}{2} \lambda_{\min}\left(\sum_{t=1}^T \BF_t\right). 
\end{align*}
Thus, $\beta_b(H) \gtrsim \lambda_{\min}(\sum_{t=1}^T \BF_t)$, which completes the proof. 
\end{proof}

We are now ready to show the correctness of our algorithm.

\begin{lemma}
If \Cref{alg:cut-matching} returns a vector $\Bx \in \{0, \pm 1\}^V$ in \Cref{line:return-x}, then $\beta_b(G) \leq \beta_b (\Bx) < r$ with probability $1$. 
If \cref{alg:cut-matching} returns a multigraph $H$ in \Cref{line:return-H}, then $\beta_b(G) = \Omega(r / \log n)$ with probability at least $1 - 1/\poly(n)$.
\end{lemma}

\begin{proof}
Suppose that the algorithm returns $\Bx$ in \Cref{line:return-x}.
This case happens only if $(A, B)$ is not $r$-well-linked.
By the proof of \Cref{thm:linked}, a minimum consistent cut $X \cup s^+$ ($X \subseteq V'$) satisfies $r^{-1} w(E(X, \ol{X})) < b(X)$ and therefore the corresponding $\Bx$ satisfies $\beta_b(\Bx) < r$.
Consequently, we have $\beta_b(G) \le \beta(\Bx) < r$. 

Assume that the algorithm returns $H$ in \Cref{line:return-H}.
We will show that $\beta_b(G) = \Omega(r / \log n)$ using \Cref{cor:congestion}.
Let us fix symmetric $(A, B)$ arbitrarily.
By $\beta_b(H) = \Omega(\log n)$ and \Cref{cor:congestion}, there exists a saturating $s^+$--$s^-$ flow in the auxiliary network of $(H', A, B)$ (with $r = 1$) with congestion $O(\log^{-1} n)$. 
Since each edge of $H'$ corresponds to a path in $\bigoplus_{t=1}^T \caP_t$, we obtain a saturating $s^+$--$s^-$ flow in the auxiliary network of $(G, A, B)$ (with $r = 1$) by rerouting.
Since the congestion of each $\caP_t$ is at most $r^{-1}$ and there are $T$ multisets, the congestion of this flow is at most $r^{-1}T \cdot O(\log^{-1} n) = O(r^{-1} \log n)$.
Since $(A,B)$ was arbitrary, $\beta_b(G) = \Omega(r\log^{-1} n)$ by \Cref{cor:congestion}.
\end{proof}

\subsection{Fast Implementation with Approximate Gram Decomposition}
Now we discuss fast implementation using approximate Gram decompositions and the time complexity of \cref{alg:cut-matching}.
We need the following lemma to compute a Gram decomposition efficiently.

\begin{lemma} [cf.~Lemma 4.18 in \cite{Lau2024}] \label{lem:app_gram}
    Let $\Bv_1, \cdots, \Bv_n$ be a Gram decomposition of the matrix $\BD_b^{-1/2} \BX_t \BD_b^{-1/2}$.
    There exists a randomized algorithm that computes vectors $\wh{\Bv}_1, \cdots, \wh{\Bv}_n \in \R^d$ for $d = O(\eps^2 \log n)$ in $O(\eps^{-2} \log n \cdot \max\{\log^2 n, \log b(V)\} \cdot \min\{b(V), n^2\})$ time such that 
    \begin{align*}
        \norm{\wh{\Bv}_i}^2 &\in (1 \pm \eps) \norm{\Bv_i}^2 \pm \frac{1}{\poly(n, b(V))} \quad (i \in V)  \\
        \norm{\wh{\Bv}_i + \wh{\Bv}_j}^2 &\in (1 \pm \eps) \norm{\Bv_i + \Bv_j}^2 \pm \frac{1}{\poly(n, b(V))} \quad (i, j \in V)
    \end{align*}
    with probability at least $1 - 1/\poly(n)$. 
\end{lemma}

The lemma follows from Johnson--Lindenstrauss dimension reduction and truncated Taylor expansion similar to \cite{Arora2016,Lau2024}.
We place a formal proof in \Cref{sec:gram} for completeness.

Now we prove the main theorem of this paper.

\begin{theorem}\label{thm:cut-matching}
    Given $r \in (0,1]$ with $1/r$ being an integer, \Cref{alg:cut-matching} finds either $\Bx \in \{0, \pm 1\}^V \setminus \{\mb{0}\}$ with $\beta_b(\Bx) < r$ or a certificate proving that $\beta_b(G) \geq \Omega(r/\log n)$ with probability at least $1 - 1/\poly(n)$.
    The time complexity is $O(\log^3 n \cdot \max\{\log^2 n, \log b(V) \} \cdot \min\{b(V), n^2\})$ arithmetic operations and $O(\log^2 n)$ single-commodity max-flow computations.
    
    By binary search on $r$, we can obtain an $O(\log n)$-approximation randomized algorithm for the $b$-bipartiteness ratio of undirected graphs.
    The time complexity is $O(\log (w(E) \cdot b(V)) \cdot \log^3 n \cdot \max\{\log^2 n, \log b(V)\} \cdot \min\{b(V), n^2\})$ arithmetic operations and $O(\log (w(E)\cdot b(V)) \cdot \log^2 n)$ single-commodity max-flow computations.
\end{theorem} 

\begin{proof}
We first prove the correctness.
We have already seen the correctness if the algorithm is given an exact Gram decomposition.
To avoid repetition, we sketch the proof for approximate Gram decompositions. 
The only part of the analysis to be changed is the bound of $\gamma$.
Let $\Bv_1, \dots, \Bv_n$ be a Gram decomposition of $\BD_b^{-1/2} \BX_t \BD_b^{-1/2}$ and $\wh{\Bv}_1, \dots, \wh{\Bv}_n$ be the approximate Gram decomposition given by \cref{lem:app_gram}.
Note that $\wt{v_i}$ is now computed based on $\wh{\Bv}_i$ instead of $\Bv_i$.
Then, by \cref{lem:app_gram}, we have
\begin{align*}
    \sum_{ij \in M} \norm{\Bv_i + \Bv_j}^2 &\in (1 \pm \eps) \sum_{ij \in M} \norm{\wh\Bv_i + \wh\Bv_j}^2 \pm \frac{1}{\poly(n, b(V))}, \\
    \sum_{i \in V} b(i) \norm{\Bv_i}^2 &\in (1 \pm \eps) \sum_{i \in V} b(i)\norm{\wh\Bv_i}^2 \pm \frac{1}{\poly(n, b(V))}.
\end{align*}
Therefore, using approximate Gram decomposition incurs at most $(1 \pm \eps)$ multiplicative error and $1/\poly(n, b(V))$ additive error, where the latter is negligible.
Therefore, if we take $\eps$ to be a small enough constant, then only $O(1)$ multiplicative error is incurred.
Thus, the desired bound $\gamma \gtrsim 1/\log n$ still holds.

Now we analyze the time complexity. 
The algorithm makes $T = O(\log^2 n)$ iterations, and each iteration requires finding an approximate Gram decomposition, computing a maximum $s^+$--$s^-$ flow in the auxiliary network and flow decomposition.
Note that a consistent minimum $s^+$--$s^-$ cut can be obtained by removing all inconsistent vertices in any minimum $s^+$--$s^-$ cut  (see \Cref{lem:cut}).
Therefore, the time complexity of finding consistent minimum cut is the same as max-flow.
Approximate Gram decomposition takes $\tilde O(\eps^{-2} \min\{b(V), n^2\}) = \tilde O(\min\{b(V), n^2\})$ time by \cref{lem:app_gram}.
Flow decomposition takes $O(m \log n)$ time~\cite{Lee2013}, which is subsumed by max-flow.
Thus, the claimed time complexity follows.
\end{proof}

\section{Application to Minimum Uncut} 
\label{sec:min-uncut}
In this section, we present an approximation algorithm for the minimum uncut problem, building on the 
$O(\log n)$-approximation algorithm for the bipartiteness ratio problem introduced in the previous section.

For a graph $G=(V,E,w)$, we define the \emph{value of a cut} $(S,\overline{S})$ as $w(E(S,\overline{S}))/w(E)$.
Similarly, the \emph{uncut value} of $(S, \overline{S})$ is defined as $(w(E(S)) + w(E(\overline{S})))/w(E) = 1 - w(E(S, \overline{S})) / w(E)$.
Generally, our method is conceptually similar to those in \cite{kllot13, Trevisan2012}, which showed a close connection between min-uncut and the bipartiteness ratio. Briefly, we invoke \cref{alg:cut-matching}, which returns a partition of the vertex set $V=L \cup R \cup Z$, and then recursively apply this process to the subgraph induced by the vertex subset $Z \subset V$, where the entries of $\Bx_Z$ are zero. 
The proposed approximation algorithm for the minimum uncut problem is presented in \cref{alg:max_cut}, and its guarantee is encapsulated in the following theorem.

\begin{algorithm}[htbp]
   \caption{$\textsc{RecursiveBipart}(G = (V, E,w))$} 
    \label{alg:max_cut} 
    Let $(L, R)$ be the output of \cref{alg:cut-matching}.
    
    \eIf{$L \cup R = V$}
    {
    \Return $(L, R)$.
    }{
    $V' \gets V \setminus (L \cup R)$.

    Let $G' = (V', E', w')$ be the subgraph of $G$ induced by $V'$.

    $(L', R') \gets \textsc{RecursiveBipart}(G')$.

    \Return the cut, either $(L \cup L’, R \cup R’)$ or $(L \cup R’, R \cup L’)$, that has the larger cut value.
}
\end{algorithm}

\begin{theorem}\label{thm:max-cut}
Suppose \cref{alg:cut-matching} outputs a $C$-approximation algorithm for the bipartiteness ratio problem. Given a graph $G=(V,E,w)$ with the min-uncut value $\eta$, $\Call{RecursiveBipart}{G}$ returns a cut with the uncut value $O(C \log(1/\eta)) \cdot \eta$ in running time $\wt{O}(m n)$. 
\end{theorem} 
Combining \Cref{thm:cut-matching}, which guarantees that \cref{alg:cut-matching} achieves an $O(\log n)$-approximation, with \Cref{thm:max-cut} yields \Cref{thm:min-uncut-intro}.

Before proving this theorem, we first introduce the following facts that will be employed in the proof of \cref{thm:max-cut}.  
\begin{fact} \label{fact:fx_mono}
(1) For $x \in (0, 1)$, $f(x) := x \log (3 / x)$ is monotonically increasing. 

(2) Given $\eta \in (0, 1)$, for any $x \in (0, 1)$, $g(x) := x + \log \frac{3 (1-x)}{\eta} \le \log \frac{3}{\eta}$. 
\end{fact} 

\begin{proof} [Proof of \cref{thm:max-cut}] 
Consider \cref{alg:max_cut}.
Given an undirected graph $G=(V,E,w)$ with the min-uncut value $\eta$, we apply the $C$-approximate algorithm for the bipartiteness ratio on $G$. 
Let $\Bx \in \{0, \pm 1\}^V$ be the returned vector, then we have $\beta(\Bx) \le C \cdot \beta(G)$. 
Let $(S, \ol{S})$ be the maximum cut of $G$. 
Note that 
\[
\frac{w(E(S)) + w(E(\ol{S}))}{w(E)} = \eta. 
\]
Let $\By$ be a $\{\pm 1\}^V$-vector such that $\By_S = \mb{1}$ and $\By_{\ol{S}} = -\mb{1}$, then we have 
\[
\beta(\By) = \frac{w(E(S)) + w(E(\ol{S}))}{w(E)} = \eta. 
\]
Note that $\beta(\By) \ge \beta(G)$ by definition.
Combining with $\beta(\Bx) \le C \cdot \beta(G)$, we obtain $\beta(\Bx) \le C \eta$.
Let the corresponding partition of $\Bx$ be $L \cup R \cup Z$, then 
\begin{align} \label{eq:beta_x}
\beta(\Bx) = \frac{2 \cdot w(E(L)) + 2 \cdot w(E(R)) + w( E(L \cup R, Z) )}{\vol(L \cup R)}.  %
\end{align}

We now prove by induction that, given a graph $G=(V,E,w)$ with the min-uncut value $\eta$, $\Call{RecursiveBipart}{G}$ returns a cut with the uncut value $C \eta \log (3/\eta)$. 

For the \textit{base case} ($Z = \emptyset$), we have 
\[
w(E(L)) + w(E(R)) \le \frac{C \eta}{2} \cdot \vol(L \cup R) = C \eta \cdot w(E),  
\]
where the inequality follows since $V = L \cup R$; the equality follows from $\vol(L \cup R) = 2 \cdot w(E)$. 
That is, the fraction of the uncut edges is $\frac{w(E(L)) + w(E(R))}{w(E)} \le C \eta$. 

For the \textit{inductive step}, the edges that are not cut in $G$ are categorized into three parts: (1) edges with both endpoints in $L$ or $R$; (2) edges between $L \cup R$ and $V'$; and (3) edges with both endpoints in $L'$ or $R'$. 
The total weights of the first and third parts are equal to $w(E(L)) + w(E(R))$ and $w(E' \setminus E(L', R'))$, respectively. 
For the second part, since we choose the cut with a larger value in \cref{alg:max_cut} from the two cuts $(L \cup L', R \cup R')$ and $(L \cup R', R \cup L')$, the weight of the edges that are not cut between $L \cup R$ and $V'$ is at most $\frac{1}{2} w(E(L \cup R, V'))$. Consequently, the weight of edges that are not cut in $G$ is at most 
\begin{align} \label{quan:induc}
w(E(L)) + w(E(R)) + \frac{1}{2} w( E(L \cup R, V') ) + w(E' \setminus E(L', R')). 
\end{align}
By the formulation of $\beta(\Bx)$ in \eqref{eq:beta_x} and $\beta(\Bx) \le C \eta$, we have 
\begin{align} \label{ineq:notcut_p1}
w(E(L)) + w(E(R)) + \frac{1}{2} w(E(L \cup R, V')) \le \frac{1}{2} C \eta \cdot \vol (L \cup R). 
\end{align}
Additionally, the inductive hypothesis gives   
\begin{align} \label{ineq:notcut_p2}
w(E'\setminus E(L', R')) \le C \eta' \log (3/\eta') \cdot w(E'),  
\end{align} 
where $\eta'$ is the the minimum uncut value in graph $G'$. 
Let $\rho := \frac{w(E) - w(E')}{w(E)}$, then we have 
\begin{align} \label{ineq:vol_LR_ub}
\vol(L \cup R) \le 2 (w(E) - w(E')) = 2 \rho \cdot w(E)   
\end{align}
and 
\begin{align} \label{eq:wE_prime}
w(E') = (1 - \rho) \cdot w(E). 
\end{align}
For the induced subgraph $G' = (V', E')$ of graph $G$, there are two cases: (1) $V' \subset S$ or $V' \subset \ol{S}$ (recall that $(S, \ol{S})$ is the maximum cut of $G$); (2) $V' \cap S \neq \emptyset$ and $V' \cap \ol{S} \neq \emptyset$. For the first case, we have $\eta \cdot w(E) \ge w(E') \ge \eta' \cdot w(E')$; for the second case, we have $\eta \cdot w(E) \ge w(E(V' \cap S)) + w(E(V' \cap \ol{S})) \ge \eta' \cdot w(E')$. As a consequence, it holds that $\eta' \cdot w(E') \le \eta \cdot w(E)$, which implies  
\begin{align} \label{ineq:eps_prime_ub}
\eta' \le \eta \cdot \frac{w(E)}{w(E')} = \frac{\eta}{1-\rho}.   
\end{align} 

Returning to the upper bound of the weight of uncut edges in \cref{alg:max_cut}, we have 
\begin{align*} 
&~ w(E(L)) + w(E(R)) + \frac{1}{2} w(E(L \cup R, V')) + w(E' \setminus E(L', R')) \\
\le &~ \frac{1}{2} C \eta \cdot \vol(L \cup R) + C \eta' \log(3/\eta') \cdot w(E') \\
\le &~ \rho C \eta \cdot w(E) + C \eta' \log(3/\eta') \cdot w(E') \\ 
\le &~ \rho C \eta \cdot w(E) + \frac{C \eta}{1 - \rho} \log \frac{3(1-\rho)}{\eta} \cdot w(E') \\
= &~ \left(\rho + \log \frac{3 (1-\rho)}{\eta}\right) C \eta \cdot w(E) \\
\le &~ C \eta \log (3/\eta) \cdot w(E), 
\end{align*} 
where the first inequality follows from \eqref{ineq:notcut_p1} and \eqref{ineq:notcut_p2}; the second inequality follows from \eqref{ineq:vol_LR_ub}; the third inequality follows from \eqref{ineq:eps_prime_ub} and \cref{fact:fx_mono} (1); the fourth step follows from \eqref{eq:wE_prime}; the fifth inequality follows from \cref{fact:fx_mono} (2).  

Thus far, we have completed the inductive proof that if the min-uncut value of $G$ is $\eta$, then 
$\Call{RecursiveBipart}{G}$ returns a cut with the uncut value $O(C \log(1/\eta)) \cdot \eta$. We now proceed to analyze its running time. After each recursion, the number of vertices decreases by at least a constant amount, ensuring that the recursion terminates after at most $O(n)$ steps. Since the dominant computation in each recursion is the invocation of \cref{alg:cut-matching}, which runs in nearly linear time with respect to the number of edges, the overall running time of the algorithm is $\wt{O}(m n)$. 
\end{proof}

\section{Directed Bipartiteness Ratio}\label{sec:directed-bipartiteness-ratio}

The goal of this section is to extend the toolkit for bipartiteness ratio to directed graphs and to prove \Cref{thm:directed-bipartiteness-ratio}.  We begin by constructing an auxiliary skew-symmetric gadget that captures directed violations in a cut-based language, yielding a ratio-cut formulation of the directed $b$-bipartiteness ratio.  
We then leverage this representation to design an $O(\log n)$ approximation via directed metric embeddings, mirroring the classic Leighton--Rao framework in the directed setting. Additionally, we extend the approximate min uncut algorithm to directed graphs following the framework shown in \cref{sec:min-uncut}. 

\subsection{Gadget and Characterization with Ratio Cut}\label{subsec:directed-characterization}
In this section, we show a characterization of directed bipartiteness ratio with directed flow and cut similar to \Cref{sec:bipartite-flow-cut}.
To this end, we introduce the following auxiliary graph $G'$.
For each edge $e = (i, j) \in E$, consider the following gadget.
The vertex set of the gadget consists of copy vertices $i^+, i^-, j^+, j^-$ and dummy vertices $k^+, k^-, \ell^+, \ell^-$.
The edge set of the gadget consists of $(i^+, \ell^-)$, $(i^-, \ell^+)$, $(\ell^-, i^+)$, $(\ell^+, i^-)$, $(j^+, \ell^+)$, $(j^-, \ell^-)$, $(\ell^+, j^+)$, $(\ell^-, j^-)$, $(k^+, \ell^-)$, $(\ell^+, k^-)$, $(\ell^+, k^+)$, and $(k^-, \ell^-)$. 
See~\Cref{fig:directed-gadget} for an illustration. Evidently, the gadget is skew-symmetric.

\begin{figure}
    \centering
    \includegraphics[scale = 0.85]{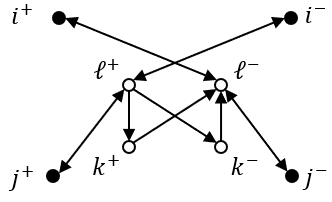}
    \caption{The gadget for directed edge $e = (i, j)$. The dummy vertices are drawn as hollow circles.}
    \label{fig:directed-gadget}
\end{figure}

Then, the auxiliary graph $G' = (V', E')$ of $G$ is defined as the digraph obtained by glueing the gadgets for all edges in $E$ along the copy vertices.
Again, $G'$ is skew-symmetric.

A subset $X \subseteq V'$ is said to be consistent to $\Bx \in \{0, \pm 1\}^V$ if $X$ is consistent and $i^+ \in X$ if and only if $x_i = 1$, and $i^- \in X$ if and only if $x_i = -1$.
\begin{lemma}\label{lem:psi-cut}
    For any $\Bx \in \{0, \pm 1\}^V$,
    \[ 
    \sum_{e=ij \in E} w(e)\psi_{ij}(\Bx) = \min_{X \subseteq V': \text{consistent to $\Bx$}} w(E'(X, \ol{X}))
    \]
\end{lemma}
\begin{proof}
    It suffices to show that for each edge $e = ij \in E$, the gadget of $e$ contributes $\psi_{ij}(\Bx)$ in any minimum cut consistent to $\Bx$. 
    This can be checked by brute force.
\end{proof}

A vertex weight $b: V \to \R_+$ is extended to $b': V' \to \R_+$ by $b'(i^+) = b'(i^-) = b(i)$ for copy vertices $i^+, i^-$ of $i \in V$ and $b'(v) = 0$ for all dummy vertices.
Abusing the notation again, we denote $b'$ by the same symbol $b$.
We also extend an edge weight $w: E \to \R_+$ to $w': E' \to \R_+$ by assigning $w'(e) = w(e)$ for edges $e$ in the gadgets corresponding to $e$.
This is also denoted by the same symbol $w$.

\begin{theorem}\label{thm:directed-repr}
    \[
        \beta_b(G) = \min_{X \subseteq V': \text{consistent}} \frac{w(E'(X, \ol{X}))}{b(X)}
    \]
\end{theorem}
\begin{proof}
    For any $\Bx \in \{0, \pm 1\}^V$, let $X \subseteq V'$ be a subset consistent to $\Bx$ with $w(E'(X, \ol{X}))$ minimum.
    Then $b(X) = \sum_{i \in V} b(i) \abs{x_i}$ and $w(E'(X, \ol{X})) = \sum_{e=ij \in E} w(e)\psi_{ij}(\Bx)$ by \Cref{lem:psi-cut}.
    This shows $\beta_b(G) \geq RHS$.

    On the other hand, take any minimizer $X \subseteq V'$ of RHS.
    Then $X$ corresponds to nonzero $\Bx \in \{0, \pm 1\}^V$.
    We have $b(X) = \sum_{i \in V} b(i) \abs{x_i}$.
    By \Cref{lem:psi-cut}, we have $w(E'(X, \ol{X})) \geq \sum_{e=ij \in E} w(e)\psi_{ij}(\Bx)$, which shows $\beta_b(G) \leq RHS$.
\end{proof}

\subsection{$O(\log n)$-approximation with Directed Metric Embedding}\label{subsec:directed-LR}
We now prove \Cref{thm:directed-bipartiteness-ratio}.
Our strategy mirrors the classical Leighton--Rao framework: we project arbitrary cuts to consistent ones, upper-bound the optimum consistent cut by a linear program over directed semimetrics, and then round an optimal fractional solution using a directed $\ell_1$ embedding.

\paragraph{Distance relaxation.}
We work over the class of \emph{weighted directed semimetrics}, i.e.,
functions $d:V'\times V'\to\R_{\ge 0}$ of the form
\[
d(p,q)=\sigma(p,q)+\pi(p)-\pi(q),
\]
where $\sigma$ is a symmetric semimetric on $V'$ (in particular, $\sigma$ obeys the triangle inequality)
and $\pi:V'\to\R$ is an arbitrary potential.
We solve the following LP:
\begin{align*}
\text{minimize}\quad & \sum_{(u,v)\in E'} w(u,v)\, d(u,v) \\
\text{subject to}\quad &
\text{$\sigma$ is a symmetric semimetric (triangle inequalities on $V'$)},\\
& d(u,v)=\sigma(u,v)+\pi(u)-\pi(v)\ge 0 \qquad(\forall u,v\in V'),\\
& \sum_{i\in V} b(i)\,\bigl(d(i^+,i^-)+d(i^-,i^+)\bigr) \ge 1.
\tag{LP$_{\mathrm{dir}}$}
\end{align*}

Let $LP$ denotes the optimal value of the above LP.
First, we check that it is indeed a relaxation.

\begin{lemma}[Validity]\label{lem:valid}
$LP \le \beta_b(G)$.
\end{lemma}

\begin{proof}
Let $S$ be any consistent set of $V'$.
Define the directed cut semimetric $d_S(p,q):=\mathbf{1}\{p\in S,\,q\notin S\}$.
It can be written as $d_S=\sigma_S+\pi_S(p)-\pi_S(q)$ with
$\sigma_S(p,q)=\tfrac12\,\mathbf{1}\{p\in S\oplus q\in S\}$ and $\pi_S(x)=\tfrac12\mathbf{1}\{x\in S\}$.
Then $d_S\ge 0$, $\sigma_S$ is a symmetric cut semimetric, and
$\sum_i b(i)\bigl(d_S(i^+,i^-)+d_S(i^-,i^+)\bigr)=b(S)$.
Scaling $d=d_S/b(S)$ makes the last constraint tight and gives objective value
$w(E'(S,\ol{S}))/b(S)$ by \Cref{thm:directed-repr}.
Minimizing over $S$ proves the claim.
\end{proof}

\paragraph{Embedding and rounding.}
Let $d^\star$ be an optimal solution to (LP$_{\mathrm{dir}}$), scaled so that the last constraint holds with equality.
By the directed $\ell_1$ weak-embedding theorem~\cite{charikar2006directed}, there exists a directed $\ell_1$ semimetric
$d'$ on $V'$ (i.e., a nonnegative conic combination of directed cut semimetrics)
such that, after a global rescaling, for all $p,q$,
\[
d'(p,q)\le d^\star(p,q),
\qquad
d'(p,q)\ge \frac{1}{D}\,d^\star(p,q) \;\text{whenever}\;  d^\star(p,q)\ge d^\star(q,p),
\]
with $D=O(\log |V'|)$.
Sample a directed cut $S\subseteq V'$ from the cut distribution of $d'$ and output the
\emph{consistent} cut $S^\dagger$ (Lemma~\ref{lem:ssym-cut}).

To avoid relying on the dominant orientation, we also consider the reversed semimetric
$\overleftarrow d(p,q):=d^\star(q,p)$, embed it similarly to obtain $\overleftarrow{d'}$, and define the
rounding distribution as the \emph{mixture} of the two cut distributions with equal probability.
For $S \subseteq V'$, let $S^\dagger$ be the consistent set obtained from $S$ by dropping both copies $\{i^+, i^-\}$ of every $i \in V$ that appear entirely in $S$.

\begin{lemma}[Expected numerator]\label{lem:num}
$\mathbb{E}\bigl[w(E'(S^\dagger,\ol{S^\dagger}))\bigr] \le\
\sum_{(u,v)\in E'} w(u,v)\, d^\star(u,v)$.
\end{lemma}

\begin{proof}
By Lemma~\ref{lem:ssym-cut}, $w(E'(S^\dagger,\ol{S^\dagger}))\le w(E'(S,\ol{S}))$.
For a single embedding (say $d'$),
\[
\mathbb{E}[w(E'(S,\ol{S}))]
=\sum_{(u,v)\in E'} w(u,v)\,\Pr[u\in S,\, v\notin S]
=\sum_{(u,v)\in E'} w(u,v)\, d'(u,v)
\le \sum_{(u,v)\in E'} w(u,v)\, d^\star(u,v).
\]
The same bound holds for $\overleftarrow{d'}$; hence it holds for the mixture.
\end{proof}

\begin{lemma}[Expected denominator]\label{lem:den}
$\mathbb{E}\bigl[b(S^\dagger)\bigr] \ge \dfrac{1}{2D}\,
\sum_{i\in V} b(i)\,\bigl(d^\star(i^+,i^-)+d^\star(i^-,i^+)\bigr)$.
\end{lemma}

\begin{proof}
For a single embedding $d'$, the event that
only one of $i^+, i^-$ is in $S^\dagger$ is the same as
$\{i^+\in S, i^-\notin S\}\dot\cup\{i^-\in S, i^+\notin S\}$.
Thus
\[
\mathbb{E}[b(S^\dagger)]
=\sum_i b(i)\,\bigl(\Pr[i^+\in S, i^-\notin S]+\Pr[i^-\in S, i^+\notin S]\bigr)
=\sum_i b(i)\,\bigl(d'(i^+,i^-)+d'(i^-,i^+)\bigr).
\]
In the mixture, for each $i$ at least one of the two orientations is ``dominant'';
therefore the expectation is at least
\(\tfrac{1}{2D}\sum_i b(i)\bigl(d^\star(i^+,i^-)+d^\star(i^-,i^+)\bigr)\).
\end{proof}

\begin{proof}[Proof of \Cref{thm:directed-bipartiteness-ratio}]
We first discuss approximation guarantee.
By Lemmas~\ref{lem:num} and \ref{lem:den} and the normalization of (LP$_{\mathrm{dir}}$),
\[
\frac{\E[w(E'(S^\dagger,\ol{S^\dagger}))]}{\E[b(S^\dagger)]}
 \le 2D\cdot
\frac{\sum_{(u,v)\in E'} w(u,v)\, d^\star(u,v)}
{\sum_i b(i)\bigl(d^\star(i^+,i^-)+d^\star(i^-,i^+)\bigr)}
 \le 2D\cdot LP.
\]
By Lemma~\ref{lem:valid}, $LP \le \beta_b(G)$, and
$|V'|=O(n + m) = O(n^2)$.
Thus there exists an outcome $S^\dagger$ with
\[
\frac{w(E'(S^\dagger,\ol{S^\dagger}))}{b(S^\dagger)}
 \le O(\log n)\cdot \beta_b(G).
\]
By \Cref{thm:directed-repr}, this yields $\Bx\in\{0,\pm 1\}^V$ with
$\beta_b(\Bx)\le O(\log n)\cdot \beta_b(G)$.

Next, we discuss the time complexity.
Solving (LP$_{\mathrm{dir}}$) takes polynomial time. 
The directed $\ell_1$ weak-embedding can be done in randomized polynomial time~\cite{charikar2006directed}.
Since the cut distribution is obtained by the standard threshold rounding, we can simply find the best cut in the support in the mixture distribution.
Thus we obtain a vector $\Bx$ whose expected ratio is within an $O(\log n)$ factor of $\beta_b(G)$.
This concludes the proof of \Cref{thm:directed-bipartiteness-ratio}.
\end{proof}

\subsection{Directed Minimum Uncut}\label{subsec:directed-min-uncut}

In this subsection we establish \Cref{thm:min-directed-uncut-intro} by extending the reduction from \emph{Minimum Uncut} in the undirected case to directed graphs.
Let $G=(V,E,w)$ be a directed graph with nonnegative weights $w(i, j)$ on arcs $(i, j) \in E$. For any $\Bx \in \{0, \pm 1\}^V \setminus \{\mb{0}\}$ with corresponding tripartition $V = L \cup R \cup Z$, the formulation of $\beta_b(\Bx)$ shown in \eqref{def:dir_beta_b_x} gives us  
\begin{align} \label{for:dir_beta_bx}
\beta_b (\Bx) = \frac{2 \cdot w(E(L)) + 2 \cdot w(E(R)) + 2 \cdot w(E (R, L)) + w (E (L \leftrightarrow Z) ) + w(E (R \leftrightarrow Z) )}{b(L \cup R)}, 
\end{align} 
where $E(R, L) = \{(u, v) \in E \mid u \in R, v \in L\}$, $E(L \leftrightarrow Z) = E(L, Z) \cup E(Z, L)$, and $E(R \leftrightarrow Z) = E(R, Z) \cup E(Z, R)$.

Let $G$ be a directed graph whose minimum uncut value is $\eta$. Applying the $C$-approximation algorithm for directed bipartiteness ratio from \Cref{subsec:directed-LR} (with $C = O(\log n)$) returns a vector $\Bx$ with $\beta_b(\Bx)\le C \cdot \beta_b(G)$. The next lemma relates $\beta_b(\Bx)$ to $\eta$.

\begin{lemma}[Link to directed min uncut] \label{lem:dir-link}
    Let $\Bx \in \{0, \pm 1\}^V \setminus \{\mb{0}\}$ be the vector returned by the $C$-approximation algorithm for the directed bipartiteness ratio of $G$.
    If the minimum directed uncut value of $G$ is $\eta$, then $\beta_b(\Bx) \le C \eta$.
\end{lemma}
\begin{proof}
    Let $(S, \ol{S})$ be a directed cut attaining the minimum uncut value $\eta$.
    Define $\By \in \{\pm 1\}^V$ by setting $y_i = 1$ for $i \in S$ and $y_i = -1$ for $i \in \ol{S}$.
    Using \eqref{for:dir_beta_bx} with $Z = \emptyset$, we obtain
    \[
    \beta_b(\By)
        = \frac{2 \cdot w(E(S)) + 2 \cdot w(E(\ol{S})) + 2 \cdot w(E(\ol{S}, S))}{2 \cdot w(E)}
        = \eta,
    \] 
    where the numerator equals twice the uncut weight of $(S,\ol{S})$.
    Since $\beta_b(G) \le \beta_b(\By)$ by definition, we have $\beta_b(G) \le \eta$.
    The $C$-approximation guarantee yields $\beta_b(\Bx) \le C \cdot \beta_b(G)$, and thus $\beta_b(\Bx) \le C \eta$.
\end{proof}

We now present the approximation algorithm in \Cref{alg:directedminuncut} for the directed minimum uncut problem.
The structure mirrors the undirected recursion, with the progress measure replaced by the directed bipartiteness ratio guarantee.

\begin{algorithm}
\caption{$\textsc{DirectedRecursiveBipart} G = (V, E, w)$}\label{alg:directedminuncut} 
    \DontPrintSemicolon
    Run the $O(\log n)$-approximation algorithm of \Cref{subsec:directed-LR} with $b(v) = d_{\mathrm{in}}(v) + d_{\mathrm{out}}(v)$ for each $v \in V$ to obtain $\Bx\in\{0,\pm1\}^V$.\;
    Set $L=\{i \mid x_i = +1\}$, $R=\{i \mid x_i=-1\}$, and $Z=\{i \mid x_i=0\}$.\;
    \lIf{$Z = \emptyset$}{\Return $(L, R)$.}
    Recurse on the subgraph $G[Z]$ to obtain $(L', R')$.\;
    \Return $(L\cup L', R\cup R')$.
\end{algorithm} 

\Cref{thm:min-directed-uncut-intro} is a direct consequence of the following theorem:
\begin{theorem}[Directed Min Uncut]\label{thm:directed-min-uncut}
    Given a directed weighted graph $G = (V, E, w)$ with $|V| = n$ and minimum uncut value $\eta$, $\Call{DirectedRecursiveBipart}{G}$ returns a cut whose uncut value is at most $O(\log n \log (1/\eta)) \cdot \eta$ in polynomial time. 
\end{theorem}

\begin{proof}
We prove, by induction on the recursion depth, that whenever \textsc{DirectedRecursiveBipart} is invoked on a graph whose minimum directed uncut value is $\eta$, the returned cut has uncut weight at most $2C\,\eta \log (3/\eta)\cdot w(E)$, equivalently an uncut fraction of $2C\,\eta \log (3/\eta)$. Here $C=O(\log n)$ is the approximation factor from \Cref{subsec:directed-LR}.

\emph{Base case.}
When the recursion stops we have $Z = \emptyset$ and the algorithm returns $(L,R)$.
Equation \eqref{for:dir_beta_bx} reduces to
\[
\beta_b(\Bx) = \frac{w(E(L)) + w(E(R)) + w(E(R,L))}{w(E)},
\]
which is exactly the uncut fraction of $(L,R)$.
By \Cref{lem:dir-link} this quantity is at most $C\eta$, matching the claimed bound for the base case.

\emph{Inductive step.}
Let $E'$ be the edge set of the recursive subgraph $G[Z]$, and let $\eta'$ denote its minimum directed uncut value.
By the induction hypothesis, the recursive call returns $(L',R')$ whose uncut weight is at most $2C\,\eta' \log (3/\eta') \cdot w(E')$.

It remains to bound the contribution from edges outside $E'$. Using \eqref{for:dir_beta_bx},
\[
w(E(L)) + w(E(R)) + w(E(R,L)) + w(E(L \cup R \leftrightarrow Z))
    \le \beta_b(\Bx)\, b(L \cup R)
    \le C \eta \cdot b(L \cup R),
\]
where $E(L \cup R \leftrightarrow Z) = E(L \leftrightarrow Z) \cup E(R \leftrightarrow Z)$ and the last inequality again invokes \Cref{lem:dir-link}.
Combining this with the recursive contribution, the total uncut weight of the final cut $(L \cup L', R \cup R')$ is at most
\begin{equation}\label{eq:dir_uncut_two_parts}
    C \eta \cdot b(L \cup R) + 2C\,\eta' \log (3/\eta') \cdot w(E').
\end{equation}

To relate the right-hand side to $w(E)$, define $\rho := \frac{w(E) - w(E')}{w(E)}$.
Since $b(v)=d_{\mathrm{in}}(v)+d_{\mathrm{out}}(v)$, each edge counted outside $E'$ contributes at most 2 to $b(L \cup R)$, implying
$b(L \cup R) \le 2 (w(E) - w(E')) = 2\rho\, w(E)$.
Moreover, because $\eta' w(E') \le \eta w(E)$, we have $\eta' \le \eta/(1-\rho)$.
Substituting these bounds into \eqref{eq:dir_uncut_two_parts} yields
\begin{align*}
    C \eta \cdot b(L \cup R) + 2C\,\eta' \log (3/\eta') \cdot w(E')
    &\le 2\rho C \eta \cdot w(E) + 2C \frac{\eta}{1-\rho} \log \frac{3(1-\rho)}{\eta} \cdot (1-\rho) w(E) \\
    &= 2C \eta \left(\rho + \log \frac{3(1-\rho)}{\eta}\right) \cdot w(E) \\
    &\le 2C \eta \log \frac{3}{\eta} \cdot w(E),
\end{align*}
where we used \Cref{fact:fx_mono} in the final inequality.
Dividing by $w(E)$ completes the inductive step and the proof.
\end{proof}

\section*{Acknowledgments}
TS is supported by JSPS KAKENHI Grant Numbers JP24K21315 and JP19K20212, and JST, PRESTO Grant Number JPMJPR24K5, Japan.
YY is supported by JSPS KAKENHI Grant Number JP24K02903 and JP22H05001. MY is supported by Japan Science and Technology Agency (JST) as part of Adopting Sustainable Partnerships for Innovative Research Ecosystem (ASPIRE), Grant Number JPMJAP2302. 

\printbibliography

\appendix
\section{Fast Computation of Approximate Gram Decomposition}\label{sec:gram}

In this section, we show how to compute an approximate Gram decomposition of a matrix in the form of $\BD^{-1/2} \BX \BD^{-1/2}$, where $\BD$ is a positive diagonal matrix and $\BX = \exp(\BA) / \tr(\exp(\BA))$ for some real symmetric $n \times n$ matrix.
The main result of this section is the following.

\begin{lemma}\label{lem:app_gram_general}
    Let $\BY = \BD^{-1/2} \BX \BD^{-1/2}$ be a real $n \times n$ symmetric matrix, where $\BD$ is a positive diagonal matrix and $\BX = \exp(\BA) / \tr(\exp(\BA))$ for some real symmetric $n \times n$ matrix $\BA$.
    Let $\Bv_1, \dots, \Bv_n \in \R^n$ be a Gram decomposition of $\BY$.
    Let $\eps, \tau \in (0,1)$ and $\lambda > 0$ such that $\norm{\BA}_2 \leq \lambda$.
    Then, there exists a randomized algorithm that computes vectors $\wh\Bv_1, \dots, \wh\Bv_n \in \R^d$ for $d = O(\eps^{-2} \log n)$ such that
    \begin{align*}
        \abs*{\norm{\wh\Bv_i}_2^2 - \norm{\Bv_i}_2^2} &\leq \eps \norm{\Bv_i}_2^2 + \tau \quad (i \in [n]), \\
        \abs*{\norm{\wh\Bv_i + \wh\Bv_j}_2^2 - \norm{\Bv_i + \Bv_j}_2^2} &\leq \eps \norm{\Bv_i + \Bv_j}_2^2 + \tau \quad (i, j \in [n])
    \end{align*}
    with probability at least $1 - 1/\poly(n)$.
    The time complexity is $O(\eps^{-2} \log n \cdot \max\{e^2\lambda, \log(n\norm{\BD}^{-1}/\tau)\} \cdot \MVP(\BA))$ time, where $\MVP(\BA)$ is the time for computing the matrix-vector product of $\BA$.
\end{lemma}

Since $\exp(\BA) = \exp(\BA/2) \exp(\BA/2)$, the task is equivalent to approximate the rows of $\frac{\BD^{-1/2}\exp(\BA/2)}{\sqrt{\tr(\exp(\BA))}}$.
To this end, we use the Johnson-Lindenstrauss (JL) dimension reduction and approximating matrix exponential with truncated Taylor series.

\begin{lemma}[Johnson--Lindenstrauss]
    Let $\BU$ be a random $d \times n$ matrix whose entry is $1/\sqrt{d}$ with probability $1/2$ and $-1/\sqrt{d}$ with probability $1/2$ independently.
    Let $\Bv_1, \dots, \Bv_n \in \R^n$ and $\Bv_i' = \BU\Bv_i$ for $i \in [n]$.
    Let $\eps \in (0,1)$, and $c > 0$.
    Then, if $d \geq \Omega(\eps^{-2} c\log n)$, we have 
    \begin{align*}
        \abs*{\norm{\Bv_i'}_2^2 - \norm{\Bv_i}_2^2} &\leq \eps \norm{\Bv_i}_2^2 \quad (i \in [n]) \\
        \abs*{\norm{\Bv_i' + \Bv_j'}_2^2 - \norm{\Bv_i + \Bv_j}_2^2} &\leq \eps \norm{\Bv_i + \Bv_j}_2^2 \quad (i, j \in [n])
    \end{align*}
    with probability at least $1 - n^{-c}$.
\end{lemma}

\begin{lemma}[{\cite[Lemma~7.4]{Arora2016}}]\label{lem:matexp}
    Let $\BA$ be a $n \times n$ real symmetric matrix, $\BU$ an $d \times n$ matrix, and $\BZ = \sum_{i=0}^k \frac{\BA^i}{i!} \BU^\top$.
    If $k \geq \max\{e^2\norm{\BA}, \log(1/\tau)\}$, then 
    $\norm{\exp(\BA)\BU^\top - \BZ} \leq \tau \norm{\exp(\BA)}\norm{\BU}$ and 
    $\norm{\exp(\BA)\BU^\top - \BZ}_F \leq \tau \norm{\exp(\BA)}\norm{\BU}_F$. 
    Furthermore, $\BU$ can be computed in $O(dk \MVP(\BA))$ time, where $\MVP(\BA)$ is the time for computing the matrix-vector product of $\BA$.
\end{lemma}

Consider $n \times d$ matrices 
\begin{align*}
    \BW = \exp(\BA/2) \BU^\top, \quad \BZ = \sum_{i=0}^k \frac{(\BA/2)^i}{i!} \BU^\top,
\end{align*}
where $\BU$ is a $d \times n$ JL projection matrix for $d = O(\eps^{-2}c\log n)$ and $k$ is a parameter in \Cref{lem:matexp}.
Let $\Bx_i$ and $\Bw_i$ ($i \in [n]$) be the rows of $\exp(\BA/2)$ and $\exp(\BA/2) \BU^\top$, respectively.
Then, by the JL lemma, we have 
\begin{align*}
    \norm{\Bw_i}_2^2 &\in  (1 \pm \eps) \norm{\Bx_i}_2^2 \quad (i \in [n]), \\
    \norm{\Bw_i + \Bw_j}_2^2 &\in (1 \pm \eps) \norm{\Bx_i + \Bx_j}_2^2 \quad (i, j \in [n])
\end{align*}
with probability at least $1 - n^{-c}$.
Therefore, we have
\begin{align}\label{eq:JL-tr}
    \tr(\BW\BW^\top) = \sum_{i=1}^n \norm{\Bw_i}_2^2 \in (1 \pm \eps) \sum_{i=1}^n \norm{\Bx_i}_2^2 
    = (1 \pm \eps) \tr(\exp(\BA))
\end{align}
with probability at least $1 - n^{-c}$.
Furthermore, considering $\BD^{-1/2} \exp(\BA/2)$ instead of $\exp(\BA/2)$, we have 
\begin{align}\label{eq:JL-scaled}
    \norm{b(i)^{-1/2}\Bw_i + b(j)^{-1/2}\Bw_j}_2^2 &\in (1 \pm \eps) \norm{b(i)^{-1/2}\Bx_i + b(j)^{-1/2}\Bx_j}_2^2 \quad (i, j \in [n])
\end{align}
with probability at least $1 - n^{-c}$, where $b(i)$ is the $i$th diagonal entry of $\BD$.

Consider $n \times n$ matrices
\begin{align*}
    \BY' = \frac{\BD^{-1/2}\BW\BW^\top \BD^{-1/2}}{\tr(\BW\BW^\top)}, \quad \BY'' = \frac{\BD^{-1/2}\BZ\BZ^\top \BD^{-1/2}}{\tr(\BZ\BZ^\top)}.
\end{align*}
We will show that $\BY'$ and $\BY''$ are good approximations of $\BY$.

\begin{lemma}\label{lem:err_Yp_Ypp}
    Let $\tau \leq \frac{1}{12n^{3/2}}$ and $\eps \leq \frac{1}{4}$.
    Conditioned on the event \eqref{eq:JL-tr}, we have $\norm{\BY' - \BY''}_2 \leq 12n^{3/2} \norm{\BD}^{-1} \tau$.
\end{lemma}
\begin{proof}
    Let $\BE = \BW - \BZ$. 
    We have
    \begin{align*}
        \norm{\BE} 
        &\leq \norm{\BE}_F \\
        &\leq \norm{\exp(\BA/2)} \cdot \norm{\BU}_F \tau \\
        &= \sqrt{n} \norm{\exp(\BA)}^{1/2} \tau,
    \end{align*}
    where we used $\norm{\exp(\BA/2)} = \norm{\exp(\BA)}^{1/2}$ and each row of $\BU$ is an unit vector.
    Therefore,
    \begin{align*}
        \norm{\BZ\BZ^\top - \BW\BW^\top} 
        &= \norm{\BE\BE^\top - \BE\BW^\top - \BW\BE^\top} \\
        &\leq \norm{\BE}^2 + 2\norm{\BE} \norm{\BW} \\
        &\leq 3n \norm{\exp(\BA)} \tau,
    \end{align*}
    where we used $\norm{\BW} \leq \norm{\exp(\BA/2)} \norm{\BU} \leq \norm{\exp(\BA)}^{1/2} \sqrt{n}$.
    Similarly, 
    \begin{align*}
        \abs{\tr(\BW\BW^\top) - \tr(\BZ\BZ^\top)} 
        &= \abs{\tr(\BE\BE^\top) - \tr(\BE\BW^\top) - \tr(\BW\BE^\top)} \\
        &\leq \norm{\BE}_F^2 + 2\norm{\BE}_F \norm{\BW}_F \\
        &\leq 3 n^{3/2} \norm{\exp(\BA)} \tau.
    \end{align*}
    Therefore, we have
    \begin{align*}
        \norm{\BY' - \BY''}
        &= \norm*{\frac{\BD^{-1/2}\BW\BW^\top \BD^{-1/2}}{\tr(\BW\BW^\top)} - \frac{\BD^{-1/2}\BZ\BZ^\top \BD^{-1/2}}{\tr(\BZ\BZ^\top)}} \\
        &\leq \norm*{\frac{\BD^{-1/2}\BW\BW^\top \BD^{-1/2}}{\tr(\BW\BW^\top)} - \frac{\BD^{-1/2}\BW\BW^\top \BD^{-1/2}}{\tr(\BZ\BZ^\top)}} + \norm*{\frac{\BD^{-1/2}\BW\BW^\top \BD^{-1/2}}{\tr(\BZ\BZ^\top)} - \frac{\BD^{-1/2}\BZ\BZ^\top \BD^{-1/2}}{\tr(\BZ\BZ^\top)}} \\
        &\leq \frac{\norm{\BD}^{-1} \cdot \norm{\BW\BW^\top}\cdot \abs{\tr(\BW\BW^\top) - \tr(\BZ\BZ^\top)}}{\tr(\BW\BW^\top)\tr(\BZ\BZ^\top)} + \frac{\norm{\BD}^{-1} \cdot \norm{\BW\BW^\top - \BZ\BZ^\top}}{\tr(\BZ\BZ^\top)} \\
        &\leq \norm{\BD}^{-1} \cdot \frac{\abs{\tr(\BW\BW^\top) - \tr(\BZ\BZ^\top)} + \norm{\BW\BW^\top - \BZ\BZ^\top}}{\tr(\BZ\BZ^\top)} \tag{$\norm{\BW\BW^\top} \leq \tr(\BW\BW^\top)$} \\
        &\leq \norm{\BD}^{-1} \cdot \frac{(3n + 3n^{3/2})\norm{\exp(\BA)}\tau}{\tr(\BW\BW^\top) - 3n^{3/2}\norm{\exp(\BA)}\tau} \\
        &\leq \norm{\BD}^{-1} \cdot \frac{6n^{3/2} \norm{\exp(\BA)}\tau}{(1- \eps - 3n^{3/2}\tau)\tr(\exp(\BA))}   \tag{by \eqref{eq:JL-tr} and $\norm{\exp(\BA)} \leq \tr(\exp(\BA))$} \\
        &\leq 12n^{3/2} \norm{\BD}^{-1} \tau. \tag{$\tau \leq \frac{1}{12n^{3/2}}$ and $\eps \leq \frac{1}{4}$}
    \end{align*}
\end{proof}

We now prove \Cref{lem:app_gram_general}.
Taking $c$ to be a large enough constant, we can ensure that the events \eqref{eq:JL-tr} and \eqref{eq:JL-scaled} hold with probability at least $1 - 1/\poly(n)$.
Assume the events hold in the following.
Note that for each $i, j \in [n]$,
\begin{align*}
    \norm{\Bv_i' + \Bv_j'}^2
    &= \frac{\norm{b(i)^{-1/2}\Bw_i + b(j)^{-1/2}\Bw_j}^2}{\tr(\BW\BW^\top)} \\
    &\in \frac{(1 \pm \eps)\norm{b(i)^{-1/2}\Bx_i + b(j)^{-1/2}\Bx_j}^2}{(1 \pm \eps)\tr(\exp(\BA))} \tag{by \eqref{eq:JL-tr} and \eqref{eq:JL-scaled}} \\
    &\in (1 \pm 2\eps) \norm{\Bv_i + \Bv_j}^2.
\end{align*}
By \Cref{lem:err_Yp_Ypp}, we have 
\begin{align*}
    \abs{\norm{\Bv_i' + \Bv_j'}^2 - \norm{\Bv_i'' + \Bv_j''}^2} 
    &= \abs{\inprod{\BE_i + \BE_j + 2\BE_{ij}, \BY' - \BY''}} \\
    &\leq \norm{\BE_i + \BE_j + 2\BE_{ij}}_F \cdot \norm{\BY' - \BY''}_F \\
    &\lesssim n^{3/2} \norm{\BD}^{-1} \tau.
\end{align*}
Thus,
\begin{align*}
    \abs{\norm{\Bv_i + \Bv_j}^2 - \norm{\Bv_i'' + \Bv_j''}^2} \leq 2\eps \norm{\Bv_i + \Bv_j}^2 + O(n^{3/2} \norm{\BD}^{-1} \tau).
\end{align*}
Resetting $\eps \gets \eps/2$ and $\tau \gets \frac{\Omega(\tau)}{n^{3/2} \norm{\BD}^{-1}}$, we have the desired bound.

\paragraph{Proof of \Cref{lem:app_gram}}
Let us apply \Cref{lem:app_gram_general} with $\BA = -\eta \sum_{s=1}^{t-1} \BF_s$, $\BD = \BD_b$, $\eps = O(1)$, and $\tau = 1/\poly(n, b(V))$.
Since $\norm{\BF_t} \leq \rho = O(1)$ and $T = O(\log^2 n)$, we have $\norm{\BA} = O(\log^2 n)$.
Furthermore, $\norm{\BD} = \max_{i \in V} b(i) \geq 1$ and therefore $\norm{\BD}^{-1} \leq 1$.
So $k = \max\{e^2 \norm{\BA}, \log(n \norm{\BD}^{-1}/\tau)\} = O(\max\{\log^2 n, \log b(V)\})$ suffices.
Since each $\BF_t$ is a demand matrix of at most $b(V)$ many paths, $\MVP(\BA) = O(\min\{b(V), n^2\})$.
This proves the lemma.

\end{document}